\documentclass[aps, prd, reprint, nofootinbib, longbibliography,floatfix, superscriptaddress,twocolumn,english]{revtex4-2}
\usepackage[utf8]{inputenc}
\usepackage{float}
\usepackage{soul}
\usepackage{amssymb,amsmath,amsfonts,amsthm,epsfig,epstopdf,array,mathrsfs}
\usepackage{braket}
\usepackage{graphicx}
\usepackage{dcolumn}

\usepackage[normalem]{ulem}
\usepackage{placeins}
\usepackage{afterpage}
\usepackage{bm}
\usepackage{hyperref}
\usepackage{comment}

\graphicspath{{Images/}}
\usepackage{tikz}
\usetikzlibrary{angles,quotes,calc,shapes.geometric,positioning,decorations.pathreplacing,decorations.pathmorphing}
\usepackage{tikz-3dplot}
\usepackage{nicematrix}
\usepackage{subcaption}
\usepackage{float} % for [H]
\usepackage{adjustbox} 
\theoremstyle{plain}
\usetikzlibrary{fadings}
\tikzfading[name=line fade, top color=transparent!90, bottom color=transparent!90, middle color=transparent!0]

\theoremstyle{definition}

\newtheorem{proposition}{Proposition}

\newsavebox\mybox

\newcommand{\blue}[1]{\textcolor{blue}{#1}}

\def\bea{\begin{eqnarray}}
\def\eea{\end{eqnarray}}
\def\be{\begin{equation}}
\def\ee{\end{equation}}

\begin{document}

\preprint{}

\title{Bipartite entanglement harvesting with multiple detectors}
 
\author{Santeri Salomaa}
\email{santeri.salomaa@helsinki.fi}
\affiliation{Department of Physics,
University of Helsinki, FI-00014 Helsinki, Finland}

\author{Esko Keski-Vakkuri}
\email{esko.keski-vakkuri@helsinki.fi}
\affiliation{Department of Physics,
University of Helsinki, FI-00014 Helsinki, Finland}

\author{Sergi Nadal-Gisbert}
\email{sergi.nadalgisbert@helsinki.fi}
\affiliation{Department of Physics,
University of Helsinki, FI-00014 Helsinki, Finland}

\begin{abstract}

We study bipartite entanglement harvesting from the quantum vacuum of a massless scalar field between two subsystems, each composed of a finite number of Unruh-DeWitt detectors. Using perturbation theory, we show that the leading-order negativity is fully determined by a submatrix of the reduced density matrix, with the submatrix dimension scaling only linearly with the number of detectors. Within this framework, we analyze how the detectors' spatial arrangement influences harvesting. For all three-detector configurations and several symmetric four-detector configurations, we derive analytic expressions for the negativity and identify the configurations that maximize it. For a linear chain, we find that the harvested entanglement scales linearly with the number of detectors. These results clarify how to arrange multiple detectors to optimize harvesting and show that increasing their number broadens the ranges of energy gaps and separations over which entanglement can be extracted from the field.

\end{abstract}

\maketitle

\section{Introduction}\label{sec:intro}

Entanglement in the vacuum state of a quantum field theory (QFT) plays a central role in many phenomena at the interface of quantum theory and relativity \cite{Witten:2018,reeh1961bemerkungen,Summers:1987fn,preskill92,Hawking2005,Braunstein2013,Martin-vennin2021,agullo22,Ribes-Metidieri:2024}. Beyond its conceptual significance, vacuum entanglement also constitutes a valuable resource for relativistic quantum information protocols \cite{Martin-Brown13,Yamaguchi-aida2020,hotta2009}. However, quantifying entanglement in QFT is highly non-trivial due to the mixed-state nature of local modes and the lack of a simple tensor-product factorization \cite{Hollands:2017dov, Agullo:2024har}. Moreover, due to the Type III von Neumann algebra structure of QFTs, all pure states are equally entangled and interconvertible by LOCC to arbitrary precision \cite{vanLuijk:2024cop}.

\textit{Entanglement Harvesting} provides an operational way to bypass these difficulties: two initially uncorrelated, localized probes interact with the field and become entangled by extracting pre-existing vacuum correlations, even across spacelike separation \cite{Valentini1991,reznik1,reznik2,Pozas-Kerstjens:2015,Pozas2016}. Harvesting protocols then give a way to indirectly probe entanglement properties of QFT states. Entanglement harvesting has been studied in a wide range of scenarios in both flat and curved spacetimes \cite{Salton:2014jaa,Ng2014,mutualInfoBH,freefall,HarvestingSuperposed,Henderson2019,bandlimitedHarv2020,ampEntBH2020,carol,boris,ericksonWhen,twist2022,SchwarzchildHarvestingWellDone}, and there is a growing effort to realize it experimentally \cite{Teixido-Bonfill:2025wqb,Lindel:2023rfi,Gooding:2023xxl,Perche:2025buo}.

Most studies of entanglement harvesting focus on setups involving only two detectors. However, the inherently multimode structure of entanglement in quantum field theory \cite{Agullo:2024har} suggests that employing multiple detectors localized in different regions of space could enhance the extraction of field quantum correlations. Our goal in this work, is to explore this possibility and investigate how the spatial arrangement of multiple detectors can affect and potentially amplify the amount of harvested entanglement.

To our knowledge, tripartite entanglement harvesting with Unruh-Dewitt detectors has been studied for three pointlike detectors in some particular configurations in Minkowski spacetime in \cite{threeHarvesting2022} and later extended to black hole spacetimes in \cite{tripartiteBHarvesting}. Bipartite entanglement harvesting with multiple qubits has been explored perturbatively in \cite{Kukita-Nambu-2017} in de Sitter spacetime, using the GKLS master equation to evolve the detector state and invoking large-separation approximations to obtain analytical expressions for the leading-order negativity. Their results show that super-horizon entanglement becomes detectable as the number of detectors grows.

In this work, we extend the study of bipartite entanglement harvesting to systems of multiple Unruh-DeWitt detectors within a perturbative framework. We first derive the detectors' reduced density matrix and show that the leading-order negativity is fully determined by a submatrix supported on the one-excitation subspace. This generalizes the result of \cite{Kukita-Nambu-2017} by imposing no approximations on detector separations. Since the size of this submatrix scales only linearly with the number of detectors, whereas the full density matrix scales exponentially, this approach is computationally efficient for large multi-detector systems. Furthermore, although negativity is not an exactly additive measure\footnote{Instead, logarithmic negativity is additive.}, we show that within a perturbative regime, its leading-order correction is additive for product states.

Using this framework, we then focus on three- and four-detector setups and optimize the detector  configurations to maximize the bipartite entanglement harvested between two causally disconnected subsystems in a Minkowski spacetime. Finally, we compute a model with a linear arrangement of detectors, and show that leading-order negativity scales linearly with the number of detectors.

Our results show that using more than two detectors enhances the harvested entanglement and enlarges the parameter regime where it can be extracted. We also find that harvesting is optimized by minimizing the distance between detectors belonging to different subsystems while maximizing the separation within each subsystem.

This manuscript is organized as follows. In Sec.~\ref{sec:set-up}, we present the framework with multiple detectors perturbatively interacting with the field.  In Sec.~\ref{sec:Perturbative evaluation of the reduced density matrix}, we compute the reduced density matrix of the detectors after the interaction. In Sec.~\ref{sec:negativity}, we introduce negativity as our entanglement quantifier and show that at leading order in perturbation theory it is additive. In Sec.~\ref{sec:method-leading-order-Neg}, we show that the submatrix required to extract the leading-order negativity scales linearly with the number of detectors. In Sec.~\ref{sec:Searching for the optimal spatial configurations}, we apply the method to different scenarios involving multiple detectors. In Sec.~\ref{sec:two-detect}, we first review the harvesting protocol for two pointlike detectors. In Sec.~\ref{sec:three-detect}, we consider three detectors and provide a systematic analysis covering all possible detector geometries compatible with the causal-separation
condition. In Sec.~\ref{sec:four-detect}, we present an extensive analysis of how
the spatial geometry of a four-detector setup influences entanglement harvesting. In Sec.~\ref{sec:Alternative linear chain}, we study a model where several detectors are arranged on a one-dimensional lattice with alternating subsystems $A$ and $B$. In the final subsection, Sec.~\ref{sec:Comparison with compactly supported switching function}, we discuss the main similarities and differences when using compactly supported switching functions instead of Gaussian functions.  Finally, we summarize our findings and discuss future research avenues in Sec.~\ref{sec: conclusions}.

Throughout this paper, we use natural units with $\hbar = c = 1$.

\section{Model of multiple Unruh-DeWitt detectors}\label{sec:set-up}

We consider a finite set of Unruh-DeWitt (UDW) detectors, interacting with a real massless quantum scalar field in a $3+1$ Minkowski spacetime. The interaction is given by the following Hamiltonian in the interaction picture
\begin{align}
    \hat{H}_I(t) =  \sum_{i=1}^{N} \lambda_i\,\chi_i(t)\,\hat{\mu}_i(t)\int_{\Sigma_t}d^3x\,f_i(\mathbf{x})\,\hat{\phi}(\mathbf{x},t), \label{eq:h_int}
\end{align}
where $N$ is the total number of detectors, $\lambda_i$ are the coupling constants, $\chi_i (t)$ are the switching functions which control the interaction time of each detector and $f_i(\mathbf{x})$ are the smearing functions, telling where the detector is located in space. The interaction with the quantum field is assumed to happen in a localized region of spacetime. The UDW detector monopole moments are denoted by $\hat{\mu}_i (t)$ and are given by
\begin{align}
    \hat{\mu}_i(t) = e^{i\Omega_i t}|1_i\rangle\langle0_i| + e^{-i\Omega_i t}|0_i\rangle\langle1_i|, \label{eq:mu}
\end{align}
where $\Omega_i$ is the energy gap of the $i^\mathrm{th}$ detector. The subscript index in the bra and ket vectors indicates the detector's Hilbert space to which it operates. The massless quantum scalar field can be expanded in plane-wave modes in the following way,
\begin{align}
    \hat \phi(\mathbf{x},t)= \int \frac{d^3k}{\sqrt{(2\pi)^{3} 2|\mathbf{k}|}}\left( \hat a_{\mathbf{k}} \varphi_k(t) e^{i \mathbf{k}\cdot\mathbf{x}}  + h.c. \right), \label{fieldmodes}
\end{align}
where $\varphi_{k}(t)=e^{-i |\mathbf{k}| t}$ and the creation and annihilation operators $\hat a^{\dagger}_{\mathbf{k}}$ and $\hat a_{\mathbf{k}}$ satisfy the usual canonical commutation relations $[\hat a_{\mathbf{k}}, \hat a^{\dagger}_{\mathbf{k'}}]= \delta^{3}(\mathbf{k}-\mathbf{k'})$.

We consider a bipartition of the set of $N$ detectors into two disjoint, non-empty subsets, $A$ and $B$, consisting of $N_A$ and $N_B$ detectors, respectively, such that $N=N_A +N_B$. The total Hilbert space then factorizes as
\begin{align*}
    \mathcal{H} &= \mathcal{H}_A \otimes \mathcal{H}_B \otimes \mathcal{H}_\phi \\
    &= \Bigg( \bigotimes_{i=1}^{N_A} \mathcal{H}_i \Bigg) \otimes \Bigg( \bigotimes_{j=N_A+1}^{N_A+N_B} \mathcal{H}_j \Bigg) \otimes \mathcal{H}_\phi \\
    &= \left( \bigotimes_{k=1}^{N_A+N_B} \mathbb{C}^2 \right) \otimes \mathcal{H}_\phi.
\end{align*}
Keeping this ordering of the Hilbert spaces in mind, the detector index $k$ uniquely determines whether it belongs to subsystem  $A$ or $B$.

\subsection{Perturbative evaluation of the reduced density matrix}\label{sec:Perturbative evaluation of the reduced density matrix}

We perform a leading-order perturbative analysis where we assume all detector couplings $\lambda_i$ to be identical and lying in the weak-coupling regime, i.e., $\lambda\equiv\lambda_i  \ll 1$. We will choose the initial separable state to be the tensor product of the vacuum state of the field and the ground state of the detectors,
\begin{align}
     \hat \rho^0 &= \hat \rho_D^0 \otimes \hat \rho_\phi^0 \label{init_rho} ,
\end{align}
where $\hat \rho_\phi^0 = | 0 \rangle\langle 0 |$ and $\hat \rho_D^0$ is the initial state of the detectors
\begin{align}
    \hat \rho_D^0 &= \bigotimes_{i=1}^{N_A+N_B}|0_i\rangle\langle 0_i| \label{init_rho_detectors} .
\end{align}

In the interaction picture, time evolution is implemented by evolving the state with the unitary 
\begin{equation}\label{eq:UI}
    \hat{U}_I = \mathcal{T}\exp\left(- i \int dt \, \hat{H}_I(t)\right),
\end{equation}
where $\mathcal{T}$ denotes the time-ordering operator, such that the final state is
\begin{align}
    \hat \rho^f &= \hat U_I\hat \rho^0 \hat U^\dag_I \label{final_rho} .
\end{align}
We expand this unitary perturbatively using the Dyson series,
\begin{align*}
    \hat U_I&= \openone+\hat U_I^{(1)}+\hat U_I^{(2)}+\mathcal{O}(\lambda^3),
\end{align*}
where the first terms are
\begin{align*}
        \hat U_I^{(1)}&=-i\int dt\,\hat  H_I(t), \\
        \hat U_I^{(2)}&=-\int dt dt'\,\mathcal{T}\{\hat H_I(t)\hat H_I(t')\} \\
        &=-\int dt dt'\,\Theta(t-t') \hat H_I(t) \hat H_I(t').
\end{align*}
Substituting these expressions into $\hat \rho^f$ in Eq.~\eqref{final_rho}, yields a perturbative expansion of the final state
\begin{align}
    \hat \rho^f &=  \bigg( \sum_i \hat U_I^{(i)}\bigg)\hat \rho^0 \bigg( \sum_j \hat U_I^{(j)\dagger}\bigg) = \sum_n \lambda^n\hat\rho^{(n)}.
\end{align}
Thus,
\begin{align*}
    \hat \rho^f &= \hat \rho^{(0)} + \lambda \hat \rho^{(1)} + \lambda^2\hat \rho^{(2)}  +\mathcal{O}(\lambda^3),
\end{align*}
where
\begin{align*}
    \hat \rho^{(0)} &= \hat \rho^0 \\
    \lambda\hat \rho^{(1)} &= \hat U_I^{(1)}\hat \rho^0+\hat \rho^0\hat U_I^{(1)\dag} \\
    \lambda^2 \hat \rho^{(2)} &= \hat U_I^{(1)}\hat \rho^0\hat U_I^{(1)\dag}+\hat U_I^{(2)}\hat \rho^0+\hat \rho^0\hat U_I^{(2)\dag}.
\end{align*}

The detectors’ final state after the interaction is obtained by tracing over the field’s degrees of freedom. This gives the reduced density matrix
\begin{align}
    \hat \rho_{AB} := \mathrm{Tr}_\phi(\hat \rho^f). \label{reduced_density_matrix_def}
\end{align}
For the vacuum state of the scalar field, all odd  $n$-point functions vanish. Hence, $\hat\rho^{(1)}_{AB}=0$, and the leading nonzero correction to $\hat \rho_{AB}$ arises at $\mathcal{O}(\lambda^2)$.

To express the perturbed state in a form suitable for later bipartite entanglement analysis, we introduce some basic definitions. First, we define index sets for the detectors corresponding to the subsystem split
\begin{align}
    A &:= \{1,2,...,N_A\} \label{index_set_A}, \\
    B &:= \{N_A+1,N_A+2,...,N_A+N_B\} \label{index_set_B}.
\end{align}
Second, following 
\cite{Kukita-Nambu-2017}, we present detector basis vectors by separating the subsystems $A$ and $B$ with a colon:
\begin{itemize}
    \item $\ket{0:0}=\bigotimes_{i=1}^{N_A+N_B} \ket{0}$: the ground state.
    \item $\ket{i:0}$: only the $i^\mathrm{th}$ ($i\in A$) detector is excited.
    \item $\ket{0:j}$: only the $j^\mathrm{th}$ ($j\in B$) detector is excited.
    \item $\ket{i:j}$: both $i^\mathrm{th}$ ($i\in A$) and $j^\mathrm{th}$ ($j\in B$) detectors are excited.
    \item $\ket{ij:0}$: both $i^\mathrm{th}$ and $j^\mathrm{th}$ detectors ($i,j\in A$) detectors are excited.
    \item $\ket{0:ij}$: both $i^\mathrm{th}$ and $j^\mathrm{th}$ detectors ($i,j\in B$) detectors are excited.
\end{itemize}
At leading order $\mathcal{O}(\lambda^2)$, these basis vectors provide the only nonzero components of the reduced density matrix $\hat \rho_{AB}$, since the detectors are initially in their ground states.

Performing the perturbative expansion, $\hat \rho_{AB}$ decomposes into two orthogonal subspaces and takes a block-diagonal form,
\begin{align}
    \hat \rho_{AB} &= \rho_1\oplus\rho_2+\mathcal{O}(\lambda^4). \label{reduced_density_matrix_direct_sum}
\end{align}
The first block, $\rho_1$, contains the one-excitation terms and is supported on the subspace spanned by the basis
\begin{align}
    \big\{ \ket{i:0},\,\ket{0:j} \,\big| \, i\in A,\,j\in B\big\}. \label{basis1}
\end{align}
Explicitly,
\begin{align}
    \rho_1 &= \sum_{i\in A}P_{i} \ket{i:0}\bra{i:0} + \sum_{i\in B}P_{i} \ket{0:i}\bra{0:i} \nonumber\\
    &\quad+ \sum_{i>j\in A}\Big( C_{ij}\ket{j:0}\bra{i:0} + C_{ij}^*\ket{i:0}\bra{j:0} \Big) \nonumber\\
    &\quad+ \sum_{i>j\in B}\Big( C_{ij}\ket{0:j}\bra{0:i}  + C_{ij}^*\ket{0:i}\bra{0:j}\Big) \nonumber\\
    &\quad+ \sum_{\substack{i\in B\\j\in A}}\Big( C_{ij}\ket{j:0}\bra{0:i}  + C_{ij}^*\ket{0:i}\bra{j:0}\Big), \label{rho1}
\end{align}
The second block, $\rho_2$, contains the vacuum and two-excitation terms and is supported on the complementary subspace,
\begin{align}
    \rho_2 &= \Big(1 - \sum_{i=1}^{N}P_{i}\Big)\ket{0:0}\bra{0:0} \nonumber\\
    &\quad+ \sum_{i>j\in A}\Big(X_{ij}\ket{ij:0}\bra{0:0}+ X_{ij}^*\ket{0:0}\bra{ij:0}\Big) \nonumber\\
    &\quad+ \sum_{i>j\in B}\Big(X_{ij}\ket{0:ij}\bra{0:0}+ X_{ij}^*\ket{0:0}\bra{0:ij}\Big) \nonumber\\
    &\quad+ \sum_{\substack{i\in B\\j\in A}}\Big(X_{ij}\ket{j:i}\bra{0:0}+ X_{ij}^*\ket{0:0}\bra{j:i} \Big)\label{rho2}.
\end{align}
The matrix elements in Eqs.~\eqref{rho1} and \eqref{rho2} are
\begin{align}
    P_{i} &:= \lambda^2\int d^4x d^4x' \Lambda_i(x)\Lambda_i(x') e^{-i\Omega_i( t- t')} W(x,x') , \label{P_term}\\
    C_{ij} &:= \lambda^2\int d^4x d^4x' \Lambda_i(x)\Lambda_j(x') e^{-i(\Omega_i t-\Omega_j t')} W(x,x') , \label{C_term}\\
    X_{ij} &:= -\lambda^2\int d^4x d^4x' \Lambda_i(x) \Lambda_j(x') e^{i(\Omega_i t+\Omega_j t')} \nonumber\\
    &\quad \times \Big(\Theta(t-t')W(x,x')+ \Theta(t'-t)W(x',x)\Big).\label{X_term} 
\end{align}
Here $\Lambda_i(x):=\chi_i(t) f_i(\mathbf{x})$ denotes the spacetime smearing function, and $W(x,x')$ is the Wightman two-point function,
\begin{align}
    W(x,x'):=\mathrm{Tr}\left( \hat \rho_\phi^0\, \hat \phi(x)\hat \phi(x')\right). \label{wightman_func_def}
\end{align}

To display a matrix representation of the reduced density matrix, we specify the ordering of the basis vectors. We adopt a block-diagonal ordering of Eq.~\eqref{reduced_density_matrix_direct_sum} which is convenient when computing the eigenvalues of the partially transposed matrix in Sec.~\ref{sec:method-leading-order-Neg}. The ordering begins with the one-excitation basis of Eq.~\eqref{basis1}, followed by $\ket{0:0}$, and then the remaining two-excitation basis, with binary ordering preserved within each block. With this ordering, $\hat \rho_{AB}$ takes the form
\begin{align}
    \hat \rho_{AB} &=
     \begin{pmatrix}
        \rho_1 & \mathbf{0} \\
        \mathbf{0} & \rho_2
    \end{pmatrix}+\mathcal{O}(\lambda^4),
\end{align}
where $\rho_1$ is the block matrix
\begin{align}
    \rho_1 &=
    \begin{pmatrix}
        \mathcal{C}_{BB} & \mathcal{C}_{BA}^\dag \\
        \mathcal{C}_{BA} & \mathcal{C}_{AA}
    \end{pmatrix}, \label{rho1_block}
\end{align}
with components explicitly given by
\begin{widetext}
\begin{align}
    \rho_1 &=
    {\footnotesize
    \begin{pNiceArray}{cccc|cccc}
        P_{N_A+N_B} & C^*_{N_A+N_B,N_A+N_B-1} & \cdots & C^*_{N_A+N_B,N_A+1}  & C^*_{N_A+N_B,N_A} & C^*_{N_A+N_B,N_A-1} & \cdots & C^*_{N_A+N_B,1} \\ 
        C_{N_A+N_B,N_A+N_B-1} & P_{N_A+N_B-1} & \cdots & C^*_{N_A+N_B,N_A+1} & C^*_{N_A+N_B-1,N_A} & C^*_{N_A+N_B-1,N_A-1} & \cdots & C^*_{N_A+N_B-1,1} \\
        \vdots & \vdots & \ddots & \vdots & \vdots & \vdots & \ddots & \cdots \\
        C_{N_A+N_B,N_A+1} & C_{N_A+N_B-1,N_A+1} & \cdots & P_{N_A+1} & C_{N_A+1,N_A}^* & C^*_{N_A+1,N_A-1} & \cdots & C^*_{N_A+1,1} \\
        \hline
        C_{N_A+N_B,N_A} & C_{N_A+N_B-1,N_A} & \cdots & C_{N_A+1,N_A} & P_{N_A} & C^*_{N_A,N_A-1} & \cdots & C^*_{N_A,1} \\ 
        C_{N_A+N_B,N_A-1} & C_{N_A+N_B-1,N_A-1} & \cdots & C_{N_A+1,N_A-1} & C_{N_A,N_A-1} & P_{N_A-1} & \cdots & C^*_{N_A-1,1} \\
        \vdots & \vdots & \ddots & \vdots & \vdots & \vdots & \ddots & \cdots \\
        C_{N_A+N_B,1} & C_{N_A+N_B-1,1} & \cdots & C_{N_A+1,1} & C_{N_A,1} & C_{N_A-1,1} & \cdots & P_{1}
    \end{pNiceArray} \label{rho1_block_opened}
    }
    ,
\end{align}
\end{widetext}
while $\rho_2$ takes the form
\begin{align}
    \rho_2 &=
    \begin{pmatrix}
        \rho_{00} & \mathbf{x}^\dag & \mathbf{0} \\
        \mathbf{x} & \mathbf{0} & \mathbf{0} \\
        \mathbf{0}& \mathbf{0}& \mathbf{0}
    \end{pmatrix}. \label{reduced_density_matrix_repr}
\end{align}
Its non-zero elements consist of the the vacuum component $\rho_{00} = 1 - \sum_i P_i$ and the vector
\begin{align}
    \mathbf{x}&=
    \begin{pmatrix}
        \mathcal{X}_{BB} \\
        \mathcal{X}_{BA} \\
        \mathcal{X}_{AA}
    \end{pmatrix}
    =
    \begin{pmatrix}
        X_{N_A+N_B,N_A+N_B-1} \\
        X_{N_A+N_B,N_A+N_B-2} \\
        \vdots  \\
        X_{21}
    \end{pmatrix} . \label{m_vector}
\end{align}

\subsection{Negativity and its leading-order additivity}\label{sec:negativity}

Characterizing bipartite entanglement between systems $A$ and $B$ after the detectors interact with the field is a difficult task. In fact, checking whether a given mixed state is separable is an NP-hard problem \cite{Gurvits:2003gdo, Gharibian:2008hgo}. Selecting an entanglement measure involves making informed compromises, as no known single measure captures all desirable features, such as computability, operational significance, monotonicity under LOCC, additivity, convexity, and asymptotic continuity.

Fortunately, some entanglement measures are both informative and easy to compute. One such a measure is the \emph{negativity}  \cite{Vidal-Werner-2002}, defined as
\begin{align}
    \mathcal{N}(\hat \rho) := \frac{\|\hat \rho^{T_B}\|_1-1}{2}
    = \sum_{\alpha_i<0} |\alpha_i|, \label{neg_definition}
\end{align}
where $\hat \rho^{T_B}$ denotes the partial transpose of the density matrix  $\hat \rho$ with respect to subsystem $B$ and $\|\cdot\|_1$ is the trace norm. This turns out to be equivalent to the sum of the absolute values of the negative eigenvalues $\alpha_i<0$ of $\hat \rho^{T_B}$. %Thus, being an entanglement measure easy to compute.

Negativity is based on the  Peres-Horodecki Positive Partial Transpose (PPT) criterion \cite{Peres:1996dw, Horodecki:1996nc}, which states that the partial transposition of a separable state $\hat\rho$ yields a positive semidefinite operator $\hat\rho^{T_B}\geq 0 $. Thus, negativity provides an efficient method for detecting entanglement, as the spectrum of the partially transposed density matrix is easily computed---unlike that of most other entanglement measures for mixed states. Moreover, it is convex function and, most importantly, an entanglement monotone under LOCC operations \cite{Vidal-Werner-2002}.

In low-dimensional bipartite systems with 
\begin{align}
    \dim(\mathcal{H}_A)\dim(\mathcal{H}_B)\leq 6, \label{dimension_boundary_on_ppt_entanglement}
\end{align}
such as in typical two Unruh-DeWitt detector entanglement harvesting models, negativity is a necessary and sufficient condition for entanglement, i.e., all PPT states are separable \cite{Horodecki:1996nc}. However, in higher-dimensional systems, there might exist entangled states with a positive partial transpose, so-called \textit{PPT entangled states}. These special states are not \textit{distillable}---meaning no maximally entangled Bell pairs can be extracted from them via LOCC \cite{Audenaert2003, entanglementmeasuresreview}---and are said to have \textit{bound} entanglement \cite{Horodecki:1998kf}. In this work we will focus on quantifying only NPT entangled states because there exist entanglement monotones for them that are both efficiently computable and reliable, while practical measures for bound (PPT) entanglement are still lacking \cite{Hiesmayr:2024cwb}. 

Since we aim to quantify entanglement with multiple detectors, the \emph{additivity} of the entanglement measure is a desirable property. Additivity means that, for a composite system of independent subsystems each containing bipartite entanglement, the total entanglement is given by the sum of the entanglement in the subsystems. The \textit{logarithmic negativity} is exactly additive, in a non-perturbative sense, whereas the standard negativity does not have this property exactly \cite{Vidal-Werner-2002}.

However, as show in Appendix \ref{App:proof_of_additivity}, the negativity is additive at leading order in perturbation theory. Consider independent bipartite subsystems
$$
\hat\rho_i = \hat\rho^{(0)}_i + \lambda^{n_i} \hat\rho^{(n_i)}_i + \mathcal{O}(\lambda^{n_i+1})
$$
with unperturbed states satisfying the PPT condition $\big(\hat\rho_i^{(0)}\big)^{T_B} \ge 0$. Then, for $n=\min_i n_i$ and sufficiently small $\lambda$, the leading-order negativity of the product state is additive:
$$
\mathcal{N}\left( \bigotimes_ {i=1}^N\hat\rho_{i} \right) = \lambda^n \sum_ {i=1}^N \mathcal{N}^{(n)}(\hat\rho_i)+\mathcal{O}(\lambda^{n+1}).
$$
The additivity follows because, under perturbation, only eigenvalues in the kernel of each subsystem’s partially transposed state can become negative, and their contributions appear at the leading order $\mathcal O(\lambda^{n_i})$.

All these facts make negativity a practical tool for quantifying bipartite entanglement in perturbative analysis. Therefore, we will use it throughout the rest of the paper.

\subsection{Extracting leading-order negativity}\label{sec:method-leading-order-Neg}

Starting from the reduced density matrix of Eq.~\eqref{reduced_density_matrix_repr}, let us derive now the leading-order negativity of the detectors after the interaction. Although the full density matrix grows exponentially with the number of Unruh-DeWitt detectors, we will show that computing the negativity only requires a block matrix whose size grows linearly. These results are consistent with the master-equation approach of \cite{Kukita-Nambu-2017}, which first established this behavior. \blue{} Here we extend that work by employing the full perturbative time evolution of the density matrix and computing the eigenvalues without additional approximations, using a matrix-based formulation.

Applying the partial transpose to the reduced density matrix $\hat\rho_{AB}$ in Eq.~\eqref{reduced_density_matrix_direct_sum} preserves its block-diagonal structure
\begin{align}
\hat \rho_{AB}^{T_B} = \tilde{\rho}_1 \oplus \tilde{\rho}_2 + \mathcal{O}(\lambda^4), \label{density_matrix_sum_pt}
\end{align}
where $\tilde{\rho}_1$ is supported on the one-excitation sector spanned by the basis vectors in Eq.~\eqref{basis1}.
\begin{align}
    \tilde{\rho}_1 &= \sum_{i\in A}P_{i} \ket{i:0}\bra{i:0} + \sum_{i\in B}P_{i} \ket{0:i}\bra{0:i} \nonumber\\
    &\quad+ \sum_{i>j\in A}\Big( C_{ij}\ket{j:0}\bra{i:0} + C_{ij}^*\ket{i:0}\bra{j:0} \Big) \nonumber\\
    &\quad+ \sum_{i>j\in B}\Big( C_{ij}\ket{0:j}\bra{0:i}  + C^*_{ij}\ket{0:i}\bra{0:j}\Big)^* \nonumber\\ 
    &\quad+ \sum_{\substack{i\in B\\j\in A}}\Big(X_{ij}\ket{j:0}\bra{0:i}+ X_{ij}^*\ket{0:i}\bra{j:0} \Big), \label{rho1_pt_dirac}
\end{align}
while $\tilde{\rho}_2$ is spanned by the remaining basis vectors:
\begin{align}
    \tilde{\rho}_2 &= \Big(1 - \sum_{i=1}^{N}P_{i}\Big)\ket{0:0}\bra{0:0} \nonumber\\
    &\quad+ \sum_{i>j\in A}\Big(X_{ij}\ket{ij:0}\bra{0:0}+ X_{ij}^*\ket{0:0}\bra{ij:0}\Big) \nonumber\\
    &\quad+ \sum_{i>j\in B}\Big(X_{ij}\ket{0:ij}\bra{0:0}+ X_{ij}^*\ket{0:0}\bra{0:ij}\Big)^* \nonumber\\
    &\quad+ \sum_{\substack{i\in B\\j\in A}}\Big(C_{ij}\ket{j:i}\bra{0:0}+ C_{ij}^*\ket{0:0}\bra{j:i} \Big)\label{rho2_pt_dirac}.
\end{align}
Comparing the $\tilde{\rho}_1$ and $\tilde{\rho}_2$ blocks to the original reduced density matrix blocks $\rho_1$ in Eq.~\eqref{rho1} and $\rho_2$ in Eq.~\eqref{rho2}, we notice that the partial transpose interchanged $C_{ij}$ and $X_{ij}$ terms with indices belonging to different subsystems and conjugated the terms having only indices of subsystem $B$.

In the block-diagonal ordered matrix representation, the partial transposition acts on $\hat \rho_{AB}$ as
\begin{align}
    \hat \rho_{AB} =
    \begin{pmatrix}
        \rho_1 & \mathbf{0} \\
        \mathbf{0} & \rho_2
    \end{pmatrix}
    \,&\longmapsto\,
    \hat \rho_{AB}^{T_B} =\begin{pmatrix}
        \tilde{\rho}_1 & \mathbf{0} \\
        \mathbf{0} & \tilde{\rho}_2
    \end{pmatrix}.\label{reduced_density_matrix_pt_block_diagonal}
\end{align}
The transformed blocks are
\begin{align}
    \tilde{\rho}_1 =
    \begin{pmatrix}
        \mathcal{C}_{BB}^* & \mathcal{X}_{BA}^{\dagger} \\
        \mathcal{X}_{BA} & \mathcal{C}_{AA}
    \end{pmatrix},\label{rho1_pt}
\end{align}
and
\begin{align}
    \tilde{\rho}_2 = \begin{pmatrix}
         \rho_{00} & \mathbf{y}^{\dagger} & \mathbf{0} \\
         \mathbf{y} & \mathbf{0} & \mathbf{0} \\
        \mathbf{0}& \mathbf{0}& \mathbf{0}
    \end{pmatrix} \label{rho2_pt}, \quad \text{where} \quad \mathbf{y}=
    \begin{pmatrix}
        \mathcal{X}_{BB}^* \\
        \mathcal{C}_{BA} \\
        \mathcal{X}_{AA}
    \end{pmatrix},
\end{align}
and $\rho_{00}=1-\sum_i P_{i}$.

We now demonstrate that the leading-order contribution to the negativity is entirely contained in the $\tilde{\rho}_1$ block in Eq.~\eqref{rho1_pt}. This is significant since, while $\hat\rho_{AB}$ grows exponentially with the number of detectors, $\tilde{\rho}_1$ grows only linearly. Explicitly, the dimension of the considered matrix reduces as:
\begin{align*}
    2^{N}\times 2^{N} \,\longrightarrow\, N\times N.
\end{align*}

To see this let us evaluate the eigenvalues of $\hat\rho_{AB}^{T_B}$. Using the block-diagonal partially transposed reduced density matrix in Eq.~\eqref{reduced_density_matrix_pt_block_diagonal}, we can split the eigenvalues $\{\alpha_i\}$ by using the determinant diagonal block rule to the characteristic equation
\begin{align*}
    \det \big(\hat \rho_{AB}^{T_B}-\alpha \openone\big)
    =  \det \begin{pmatrix}
        \tilde{\rho}_1-\alpha \openone & \mathbf{0} \\
        \mathbf{0} & \tilde{\rho}_2-\alpha\openone
    \end{pmatrix} &=0 \\
    \iff \det \big(\tilde{\rho}_1-\alpha\openone\big) \det \big(\tilde{\rho}_2-\alpha\openone\big)&=0 . 
\end{align*}
Next, using the Schur complement rule\footnote{When $D$ is invertible, $\det \begin{pmatrix}
    A & B \\
    C & D
\end{pmatrix} = \det (D) \det (A-BD^{-1}C)$} to $\det \big(\tilde{\rho}_2-\alpha\openone\big)$, we see that $\tilde\rho_2$ block in Eq.~\eqref{rho2_pt} gives always non-negative eigenvalues at the leading-order $\mathcal{O}(\lambda^2)$:
\begin{align*}
    \alpha &= 0 \quad \text{with multiplicity $2^{N}-N-2$} , \\
    \alpha &= \frac{1}{2}\left(\rho_{00}+ \sqrt{\rho_{00}^2+4\mathbf{y}^\dag\mathbf{y}}\right) \geq 0, \\
    \alpha &= \frac{1}{2}\left(\rho_{00} - \sqrt{\rho_{00}^2+4\mathbf{y}^\dag\mathbf{y}}\right) = 0 + \mathcal{O} (\lambda^4).
\end{align*}
Therefore, the $\tilde{\rho}_2$ does not contribute to the leading-order negativity.

As a result, the calculation of negativity reduces to evaluating the negative eigenvalues of the block $\tilde{\rho}_1$. Explicitly, the negativity is given by
\begin{align}
    \mathcal{N}(\hat\rho_{AB}) = \sum_{i=1}^{N} \max (0,-\alpha_i) + \mathcal{O}(\lambda^4), \label{negativity_reduction_result}
\end{align}
where $\{\alpha_i\}$ denotes the spectrum of the $\tilde{\rho}_1$ block, given in Eq.~\eqref{rho1_pt}, for an arbitrary numbers of detectors in both subsystems. Thus, although the full density matrix of the detectors grows exponentially with the number of detectors, this result shows that only a linearly growing matrix $\tilde{\rho}_1$ is needed to extract the leading-order negativity, making the entanglement analysis computationally tractable.

\section{Searching for the optimal spatial configurations}\label{sec:Searching for the optimal spatial configurations}

In the remainder of the paper, we apply the method introduced above and compute the leading-order negativity for various scenarios involving multiple detectors.

Our main objective is to identify the \emph{optimal spatial arrangements} of the detectors that maximize the genuine entanglement harvested. We focus on three and four detector cases. Finally, we demonstrate a linear chain of detectors and show that the harvested leading-order negativity increases linearly with the number of detectors.

In the next subsection we fix the smearing and switching functions used in the main analysis. Afterward, in Sec.~\ref{sec:Comparison with compactly supported switching function}, we compare those results with ones obtained using strictly compactly supported switching functions.

\subsection{Setup: switching, smearing, and causal disconnection}\label{sec:setup_and_switching}

We study how the detectors' spatial positions affect entanglement harvesting in a multi-detector setup. To isolate harvesting from signalling effects, we impose causal constraints: distinct subsystems are required to remain spacelike separated, while detectors within the same subsystem can be in causal contact \cite{martin-martinez2015,mariaPipoNew,Tjoa-Martin2021}.

To make the causal structure explicit, we decompose the Wightman function as
\begin{align}
    W(x,x') = \underbrace{ \tfrac12\langle\{\hat\phi(x),\hat\phi(x')\}\rangle}_\text{symmetric (+)} 
    + \underbrace{\tfrac12\langle[\hat\phi(x),\hat\phi(x')]\rangle}_\text{antisymmetric (-)} . \label{commutator_anticommutator}
\end{align}
This induces a corresponding decomposition of the nonlocal reduced density matrix elements (see Eqs.~\eqref{C_term}, \eqref{X_term}),
\begin{align*}
    C_{ij}=C_{ij}^++C_{ij}^-,\quad X_{ij}=X_{ij}^++X_{ij}^-.
\end{align*}
In the $\tilde\rho_1$ block of Eq.~\eqref{rho1_pt}, which determines the leading-order negativity, the terms $X_{ij}$ contain all cross-subsystem contributions. As introduced in \cite{Tjoa-Martin2021,Zambianco:2024dqn}, we can associate the symmetric part $X_{ij}^{+}$ to genuine entanglement harvesting, while the state-independent antisymmetric part $X_{ij}^{-}$ describes field-mediated communication. When dealing with more than one detector per subsystem, as in our case, one needs to analyze also the causal structure of the $C_{ij}$ terms, which can also contribute to the negativity (see for example Eqs.~\eqref{AABB_neg}-\eqref{neg_modTetrahedron}). These terms account for correlations between detectors within the same subsystem. For spacelike-separated interactions, microcausality implies that the field-commutator vanishes, so $C_{ij}^-=X_{ij}^-=0$ has to be satisfied and no causal signalling occurs.

In this analysis, we consider a simple model of stationary and identical pointlike detectors
\begin{align}
    f_i(\mathbf{x}) &= \delta^{(3)}(\mathbf{x}-\mathbf{x}_i),\label{pointlike_smearing}
\end{align}
each having the same energy gap $\Omega_i\equiv\Omega$. For the switching, we use Gaussian functions and assume that all detectors interact for the same duration, $\chi_i(t)\equiv \chi(t)$, with
\begin{align}
    \chi(t) &= \frac{1}{\sqrt{2\pi \sigma^2}} e^{-\frac{t^2}{2\sigma^2}}, \label{gaussian_switching}
\end{align}
where the interaction time is set by $\sigma$.

\begin{figure}[htbp]
\centering
\begin{tikzpicture}[scale=0.95]

    % Axes
    \draw[->, draw opacity=0.3] (-3,0) -- (3,0) node[right] {$x$};
    \draw[->, draw opacity=0.3] (0,-2.5) -- (0,3) node[above] {$t$};

    % Ticks on t-axis
    \draw[draw opacity=0.3] (-0.1,1.75) -- (0.1,1.75);
    \node[right] at (0,1.75) {$T$};
    \draw[draw opacity=0.3] (-0.1,-1.75) -- (0.1,-1.75);
    \node[right] at (0,-1.75) {$-T$};

    % Worldline A
    \draw[draw=black, line width=1.5pt, path fading=line fade] (-1.75,-1.75) -- (-1.75,1.75);
    \node[above] at (-1.75,1.8) {detector $i$};

    % Worldline B
    \draw[draw=black, line width=1.5pt, path fading=line fade] (1.75,-1.75) -- (1.75,1.75);
    \node[above] at (1.75,1.8) {detector $j$};

    % Spacelike separation indicator
    \draw[<->, dashed] (-1.65,-0.3) -- (1.65,-0.3);
    \node[below, align=center] at (0,-0.3) {minimal causal\\disconnection $L\equiv 2T$};

\end{tikzpicture}
\caption{Spacetime diagram of two detectors with minimally spacelike-separated interaction regions. Minimal causal disconnection occurs at $x_{ij}=L\equiv2T$, exact for compactly supported switching and effective for Gaussian switching with appropriately chosen parameters.}
\label{fig:spacelike_worldlines}
\end{figure}
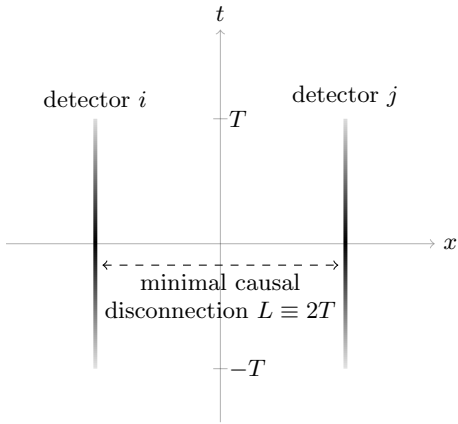

\begin{figure*}[htbp]
    \centering
    \begin{minipage}[c]{0.48\linewidth}
        \centering
        \includegraphics[width=\textwidth]{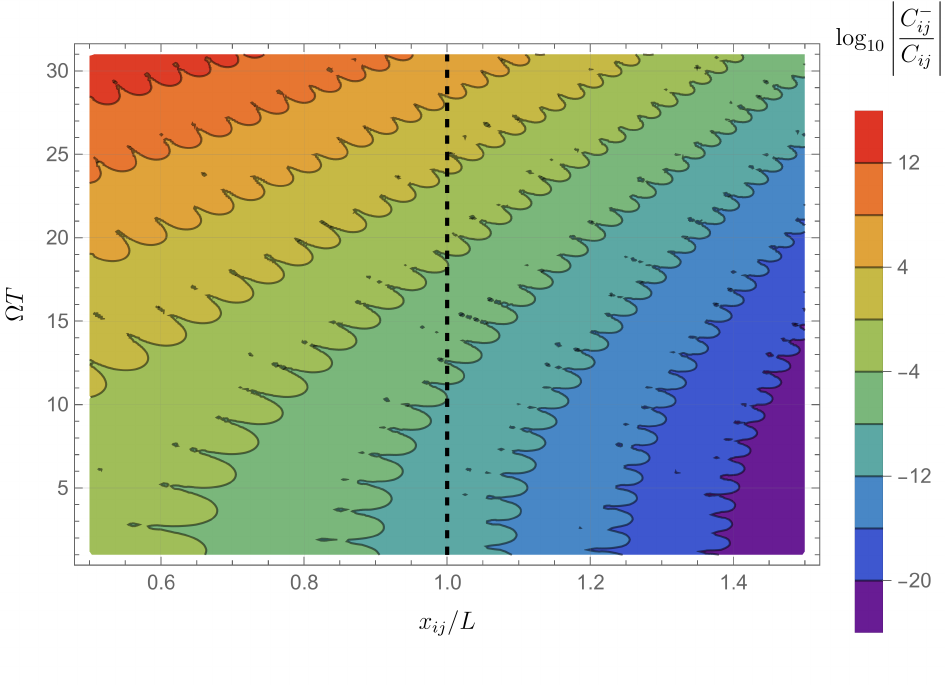}
    \end{minipage}
    \hfill
    \begin{minipage}[c]{0.46\linewidth}
        \centering
        \includegraphics[width=\textwidth]{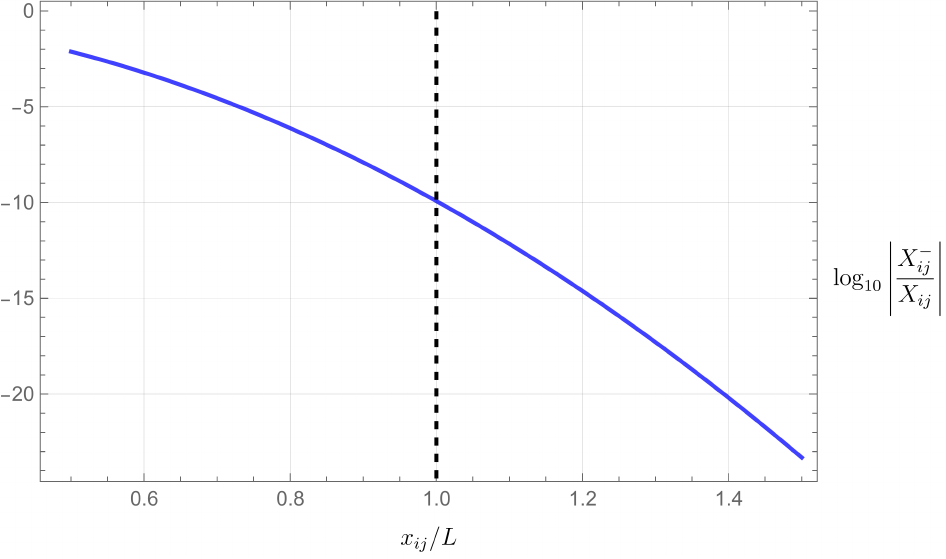}
    \end{minipage}
    \caption{Relative commutator contributions to the reduced density-matrix elements for pointlike detectors with Gaussian switching ($T=5\sigma$), illustrating effective causal disconnection at detector separation $x_{ij}=L\equiv2T$. The ratio $C_{ij}^-/C_{ij}$ oscillates around zero and depends on both $x_{ij}$ and energy gap $\Omega$, whereas $X_{ij}^-/X_{ij}$ depends only on $x_{ij}$. Near approximate causal disconnection the anticommutator terms $C_{ij}^+$ and $X_{ij}^+$, corresponding to genuine entanglement harvesting, dominate. At large $\Omega$, however, the commutator contribution $C_{ij}^-$, associated with field-mediated signalling, remains non-negligible.}
    \label{fig:CX}
\end{figure*}

The main advantage of using pointlike detectors and Gaussian switching functions is that we can obtain closed-form expressions for the density matrix elements. In Appendix \ref{app:computation_of_CX_terms}, we show that for detectors separated by $x_{ij}:=|\mathbf{x}_i-\mathbf{x}_j|$, we obtain
\begin{align}
    &\quad\,\,\,\, P = \tfrac{\lambda^2}{8\pi^2\sigma^2}\left( e^{-\Omega^2\sigma^2}-\sqrt{\pi} \Omega\sigma \,\text{erfc}(\sigma\Omega) \right), \label{P_expression}\\
    &\begin{cases}
    C_{ij}^+ = \tfrac{\lambda^2}{8\pi^{3/2}x_{ij}\sigma}e^{-\big(\tfrac{x_{ij}}{2\sigma}\big)^2}\text{Im}\left[ e^{i\Omega x_{ij}}\operatorname{erf}(i\tfrac{x_{ij}}{2\sigma}+\Omega\sigma) \right] , \\
    C_{ij}^- = -\tfrac{\lambda^2}{8\pi^{3/2}x_{ij}\sigma}\,e^{-\big(\tfrac{x_{ij}}{2\sigma}\big)^2}\,\sin(\Omega x_{ij}) ,
    \end{cases} \label{C+-_expressions}\\
    &\begin{cases}
    X_{ij}^+ = -\tfrac{\lambda^2}{4\pi^2 x_{ij}\sigma}\,e^{-\sigma^2\Omega^2}\,D\left(\tfrac{ x_{ij}}{2\sigma}\right) , \\
    X_{ij}^- = \tfrac{i\lambda^2}{8\pi^{3/2}x_{ij}\sigma}\,e^{-\sigma^2\Omega^2-\big(\tfrac{x_{ij}}{2\sigma}\big)^2} .
    \end{cases} \label{X+-_expressions}
\end{align}
Here, $\text{erf}(z)$ and $\text{erfc}(z)$ denote the error function and complementary error function and $D(z)$ denotes the Dawson integral function\footnote{\label{footnote_erf_Daw_func} The error function, complementary error function and the Dawson integral function are defined by
\begin{align*}
    \text{erf}(z) &= \frac{2}{\sqrt{\pi}} \int_0^z e^{-t^2} dt ,\quad
    \text{erfc}(z) = 1 - \text{erf}(z) , \\
    D(z) &= e^{-z^2}\int_0^z e^{t^2}dt.
\end{align*}
}. 

The downside of using  Gaussian switching functions is its lack of compact support in time, which we have to address. For these functions the conditions $C_{ij}^-=X_{ij}^-=0$ for causally disconnected detectors are satisfied only approximately, since Gaussians have infinite tails.

As argued in \cite{hectorMass,mariaPipoNew}, for non-compactly supported detectors, whether the subsystems are effectively causally disconnected must be assessed for each configuration, depending on the detector and field parameters. In our setup, this assessment depends on the detectors’ energy gaps and spatial arrangement, which we vary.

In the configurations that we consider, we can treat the interaction as effectively compact on the interval $[-T,T]$ with $T=5\sigma$, and define pointlike detectors $i$ and $j$ to be effectively causally separated when $x_{ij}=L\equiv 2T$ (see Fig.~\ref{fig:spacelike_worldlines}) for all the energy gaps considered. Figure \ref{fig:CX} shows that already for $x_{ij}\gtrsim L$ we have $X_{ij}\simeq X_{ij}^+$  and $C_{ij}\simeq C_{ij}^+$, except at large energy gaps where the contribution of $C_{ij}^-$ becomes non-negligible.

For the leading-order negativity evaluations presented in the following sections, we have verified that the $C_{ij}^-$ does not affect the optimization results where $x_{ij}>L$. This is because near optimum the same-subsystem detectors are sufficiently separated relative to their energy gap, so that replacing $C_{ij}$ with $C_{ij}^+$ leaves the leading-order negativity unchanged.

Moreover, in Sec.~\ref{sec:Comparison with compactly supported switching function} we compare our results with those obtained using strictly compactly supported switching functions, which are free from signalling effects. This comparison confirms that the Gaussian tails do not alter the optimization results. 

Throughout the following sections, results are presented using the dimensionless combinations $\Omega T$ and $x_{ij}/L$. By rescaling $t = T\tau$ in Eqs.~\eqref{P_term}-\eqref{X_term}, it follows that the density matrix elements ($P_i$, $C_{ij}$, $X_{ij}$) and the resulting leading-order negativity depend on the time scale $T$ only through these variables. Consequently, the reduced density matrix remains invariant if $T$ is varied while $\Omega T$ and $x_{ij}/L$ are held fixed. For numerical stability, we set $T = 0.01$ throughout our evaluations. 

\subsection{Two detectors}\label{sec:two-detect}

In this section, we revisit the entanglement harvesting protocol for two Unruh-DeWitt detectors, a scenario that has been extensively studied in the literature \cite{Valentini1991,reznik1,reznik2, Pozas-Kerstjens:2015, Pozas2016, Salton:2014jaa,Ng2014,mutualInfoBH,freefall,HarvestingSuperposed,Henderson2019,bandlimitedHarv2020,ampEntBH2020,carol,boris,ericksonWhen,twist2022,SchwarzchildHarvestingWellDone}. For two-qubit systems, the PPT criterion ensures that the state is separable if and only if the negativity vanishes. Building on the framework developed in Sec.~\ref{sec:method-leading-order-Neg}, we compute the leading-order negativity for this two-detector setup.

\begin{figure*}[htbp]
    \centering
    \begin{minipage}[c]{0.52\linewidth}
        \centering
        \includegraphics[width=\textwidth]{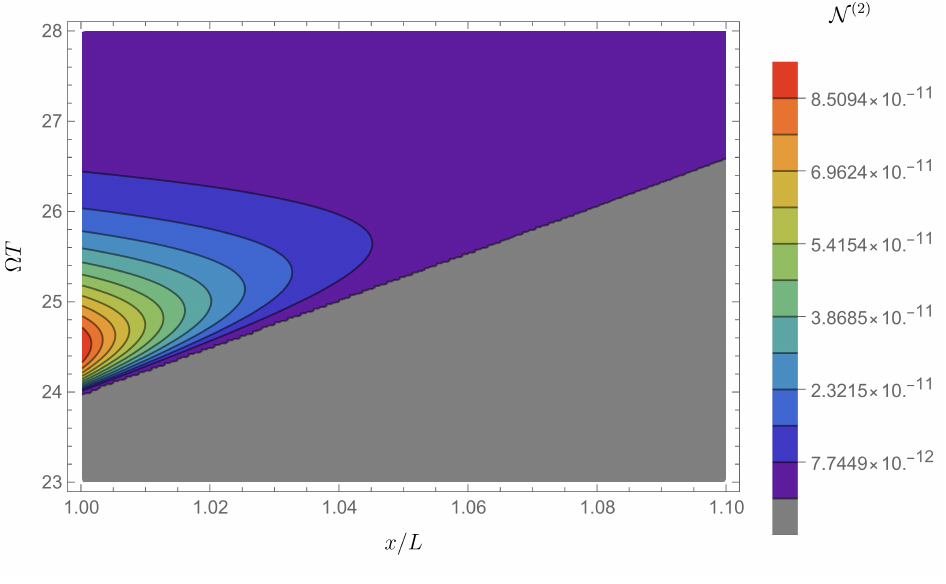}
    \end{minipage}
    \hspace{0.5cm}
    %\hfill
    \begin{minipage}[c]{0.43\linewidth}
        \centering
    \includegraphics[width=\textwidth]{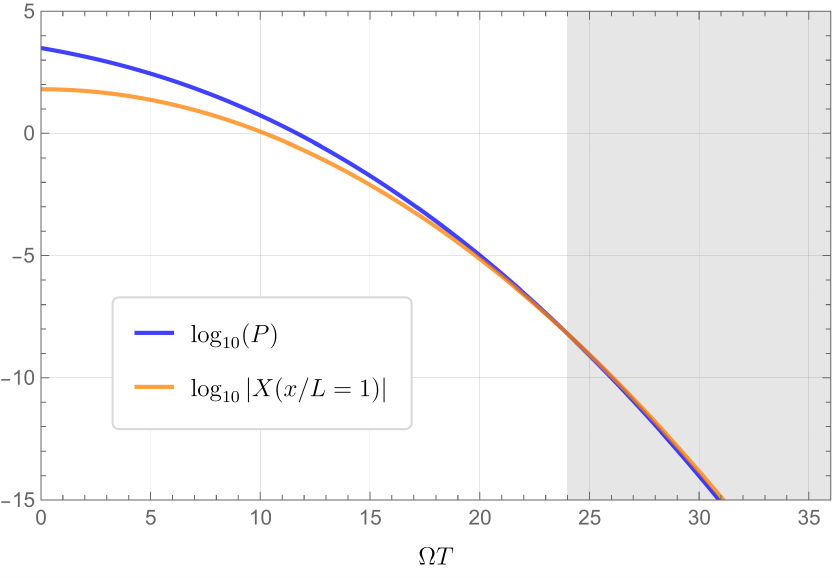}
    \end{minipage}
    \caption{The left panel shows a contour plot of the leading-order negativity between two pointlike Unruh-DeWitt detectors as a function of detector separation $x_{21}/L\equiv x/L$ (with $L\equiv 2T$ the minimal causal separation) and energy gap $\Omega T$, using Gaussian switching with $T=5\sigma$. The maximum occurs at $x/L=1$ and $\Omega T=24.49$, and the gray region indicates zero negativity. The right panel shows the magnitudes of the corresponding reduced density matrix elements, where the light gray region satisfies $|X|>P$, corresponding to non-zero entanglement harvesting.}
    %$\Omega T = 24.4938$
    \label{fig:two-detectors}
\end{figure*}

From $\tilde{\rho}_1$ block according to Eq.~\eqref{rho1_pt}, with identical detectors $P_2=P_1\equiv P$ and $X_{21}\equiv X_{x}$ separated by $x\equiv x_{21}$, we obtain a simple $2\times 2$ matrix,
\begin{align*}
    \tilde{\rho}_1 &=
    \begin{pmatrix}
        \mathcal{C}_{BB}^* & \mathcal{X}^\dag_{BA} \\
        \mathcal{X}_{BA} & \mathcal{C}_{AA}
    \end{pmatrix}
    =
    \begin{pmatrix}
        P_2 & X^*_{21} \\
        X_{21} & P_1
    \end{pmatrix}
    =
    \begin{pmatrix}
        P & X^*_x \\
        X_x & P
    \end{pmatrix}.
\end{align*}
This block has at most one negative eigenvalue, and the negativity for identical detectors is given by
\begin{align}
    \mathcal{N}(\hat\rho_{AB}) &= \max\{0,|X_x|-P\}+\mathcal{O}(\lambda^4) . \label{neg_2detectors}
\end{align}
For pointlike detectors with Gaussian switching, and under the condition that detectors in different subsystems remain effectively causally disconnected, we can substitute $P$ and $X \approx X^{+}$ using the expressions in \eqref{P_expression} and \eqref{X+-_expressions}. Combining these results, we obtain the following closed-form expression for the leading-order negativity
\begin{align}
    \mathcal{N}(x,\Omega) 
    &= 
    \max\bigg\{ 0,\, \frac{\lambda^2 e^{-\sigma^2 \Omega^2}}{8\pi^2 \sigma^2} \Big(\,
    \frac{2\sigma}{x}\, D\left( \frac{x}{2\sigma} \right) \nonumber \\
    &\quad + e^{\sigma^2 \Omega^2} \sqrt{\pi} \sigma \Omega\, \mathrm{erfc}(\sigma \Omega)-1 \Big)
    \bigg\} + \mathcal{O}(\lambda^4).\label{Negativity2detectors}
\end{align}
This explicitly depends on the distance between the detectors $x\equiv x_{21}$ and their energy gap $\Omega$. It is straightforward to check that $\frac{\partial\mathcal{N}}{\partial x}\leq 0$, so genuine entanglement harvesting maximum is at the minimum causal disconnection length $x/L=1$.

The leading-order negativity for the two-detector case is shown in Fig.~\ref{fig:two-detectors} as a function of the detector separation and energy gap. A clear maximum occurs at $\Omega T=24.49$ and $x/L=1$. From the right panel, we also see that a non-zero harvesting regime ($|X|>P$) can be achieved by increasing the energy gap.

\subsection{Three detectors}\label{sec:three-detect}

Previous work has used the $\pi$-tangle measure to study genuine tripartite entanglement harvesting with Unruh-DeWitt detectors \cite{Mendez-Henderson-Yoshimura-Mann-2022}. In particular, they analyzed three different detector geometries---linear, equilateral triangular, and scalene triangular---and found that the linear configuration allows for greater entanglement extraction among them. They also showed that tripartite entanglement can be harvested over a wider range of parameters than in the standard two-detector setup.

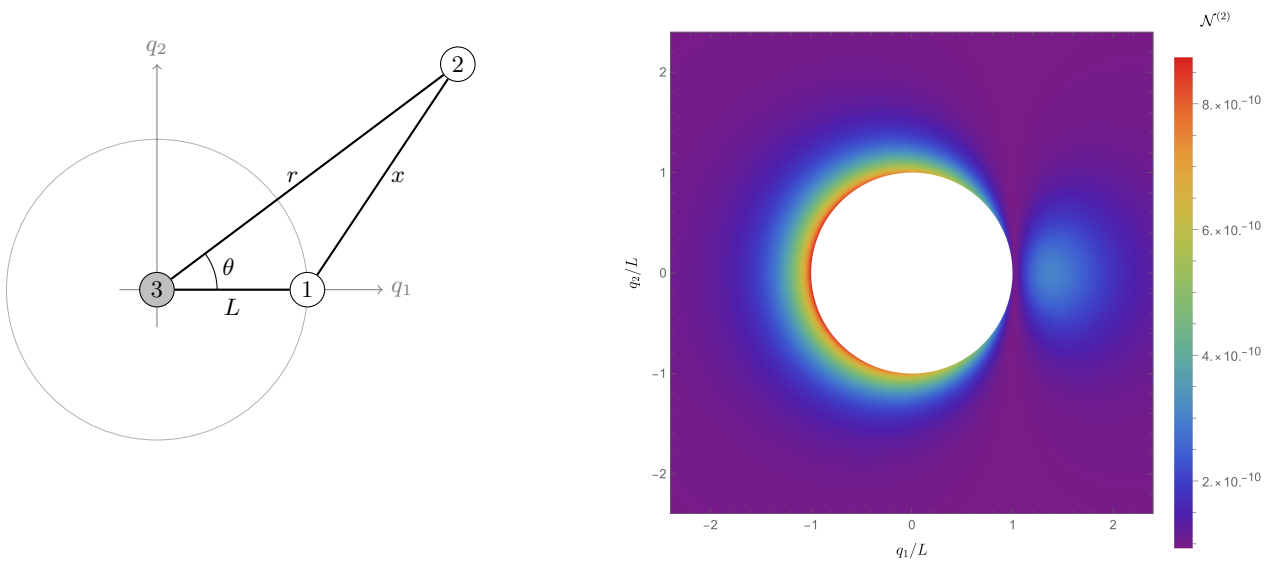
\begin{figure*}[htbp]
    \centering
    \begin{minipage}[c]{0.48\linewidth}
        \centering
        \begin{tikzpicture}
            \draw[->, gray] (-0.5,0) -- (3,0) node[right] {$q_1$};
            \draw[->, gray] (0,-0.5) -- (0,3) node[above] {$q_2$};

            % Draw a large circle around the node with faded color
            % Calculate the distance from (0,0) to node A
            \draw[opacity=0.3] (0,0) circle (2);   
        
            % Define nodes at specific coordinates
            \node (3) at (0,0) [draw, circle, fill=gray!50, inner sep=2pt] {3};
            \node (1) at (2,0) [draw, circle, fill=white, inner sep=2pt] {1};
            \node (2) at (4,3) [draw, circle, fill=white, inner sep=2pt] {2};
            \node (t) at (0.95,0.3) {$\theta$};
            
            % Draw connecting lines
            \draw[thick] (3) -- (1) -- (2) -- (3);
        
            % Add theta at vertex A
            \pic [draw, angle radius=8mm] {angle=1--3--2};
        
            % Add side labels
            \node[below] at ($(3)!0.5!(1)$) {\( L \)};
            \node[right] at ($(1)!0.5!(2)$) {\( x \)};
            \node[left] at ($(3)!0.5!(2)$) {\( r \)};
        \end{tikzpicture}
        \vspace{2cm}
        \label{fig:arrangements3}
    \end{minipage}
    \hfill
    \begin{minipage}[c]{0.48\linewidth}
        \centering
        \includegraphics[width=\textwidth]{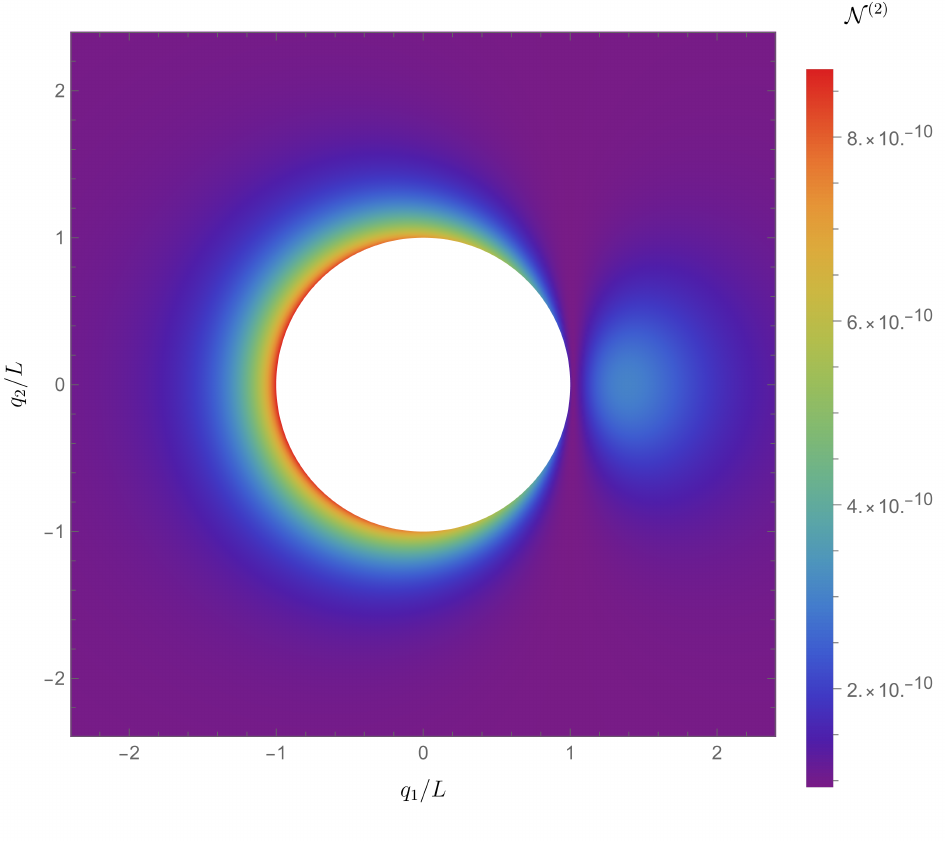}
        \label{fig:3detectors-density}
    \end{minipage}
\caption{
On the left is an illustration of the general three-detector arrangement in $\mathbb{R}^2$, with detectors 3 and 1 fixed at $(0,0)$ and $(L,0)$, respectively, where $L\equiv 2T=10\sigma$ denotes the minimal causal separation, while the position $(q_1,q_2)$ of detector 2 is varied. On the right is the corresponding density plot of the negativity, computed using the optimal two-detector energy gap ($\Omega T = 24.49$) as a representative example, where we have nonzero negativity for all positions of detector 2 with two clear local maxima.}
\label{fig:arrangements3detectors}
\end{figure*}

In this section, we carry out a full bipartite entanglement harvesting analysis between one detector and the remaining two, which constitutes the main contribution to the tripartite entanglement in the $\pi$-tangle measure.
We provide a systematic analysis covering all possible detector geometries compatible with the causal-separation condition. With this, we are able to demonstrate that the linear arrangement maximizes the entanglement harvested among all possible geometries. This result complements the ones obtained in \cite{Mendez-Henderson-Yoshimura-Mann-2022}.

To completely specify the position of the three detectors, one needs three parameters. However, by fixing one cross-subsystem pair separation to the minimal value allowed by the causal disconnection condition, only two parameters remain to be varied. This is motivated by the decaying behavior of entanglement with the distance between subsystems \cite{ubiquitous, Pozas-Kerstjens:2015}. We expect maximal entanglement when the detectors configuration is as compact as possible, with at least one cross-subsystem pair placed right at the causal boundary.

The setup is depicted in the left panel of Fig.~\ref{fig:arrangements3detectors}. Detectors 1 and 3 are fixed at the minimal causal separation $l$, while the position of detector 2 is varied. Its location is parameterized in polar coordinates $(r,\theta)$, with the origin placed at the position of detector 3. Owing to causal disconnection and the symmetries of the configuration, all inequivalent arrangements are captured by allowing $r$ to range from the minimal causal separation to infinity, and $\theta$ from $0$ to $\pi$.

For this setup, with the first two detectors in subsystem $A$ and the third detector in subsystem $B$, we can write $\tilde{\rho}_1$ in Eq.~\eqref{rho1_pt} by using distance-dependent indexing as
\begin{align}
    \tilde{\rho}_1 =
    \begin{pmatrix}
        P_3 & X^*_{32} & X^*_{31} \\
        X_{32} & P_2 & C^*_{21} \\
        X_{31} & C_{21} & P_1
    \end{pmatrix} %\nonumber\\
    &=
    \begin{pmatrix}
        P & X^*_{r} & X^*_{L} \\
        X_{r} & P & C^*_{x} \\
        X_{L} & C_{x} & P
    \end{pmatrix}. \label{3detectors}
\end{align}
Here, we have identical detectors $P_3=P_2=P_1\equiv P$, distance $r$ between detectors 2 and 3, and distance
$$
x(r,\theta)=r^2+L^2-2Lr\cos(\theta)
$$
between detectors 1 and 2 as illustrated in the left panel of Fig.~\ref{fig:arrangements3detectors}.

The leading-order negativity is given by the sum of the negative eigenvalues $\alpha_k$ of the $3\times3$ matrix in Eq.~\eqref{3detectors}. Shifting $\alpha=y+P$ reduces the characteristic equation to the depressed cubic
\begin{align*}
    y^3-S^2y-2R=0,
\end{align*}
with

\begin{align*}
    S \equiv \sqrt{|C_x|^2+|X_r|^2+|X_L|^2}, \quad R \equiv \mathrm{Re}(C_xX_rX_L^*) .
\end{align*}
Defining
\begin{align*}
    \rho &= \sqrt{\frac{4}{3}}S,\quad  \phi = \frac13 \arccos\left(\frac{3\sqrt3 R}{S^{3}}\right),
\end{align*}
the three eigenvalues can be written explicitly as
\begin{align}
    \alpha_k = P + \rho\cos\left( \phi + \frac{2\pi(k-1)}{3} \right), \quad k=1,2,3, \label{3detectors_eigenvalues}
\end{align}
and we get the negativity according to Eq.~\eqref{neg_definition} as the sum of the negative eigenvalues. Since the cosine terms are evenly spaced by $2\pi/3$, at most two eigenvalues can be negative simultaneously.

From Eq.~\eqref{3detectors_eigenvalues}, we see that the necessary and sufficient condition for vanishing leading-order negativity is determined by the smallest eigenvalue being non-negative:
\begin{align*}
    \mathcal{N}^{(2)}=0 \iff \left|\min_{k=1,2,3}\left\{\cos\left( \phi + \frac{2\pi(k-1)}{3} \right)\right\}\right|\leq P/\rho.
\end{align*}

Next, we perform a maximization analysis for the entanglement harvested where both the spatial arrangement and the energy gap are varied. The right panel in Fig.~\ref{fig:arrangements3detectors} shows the leading-order negativity as a function of all the possible positions of detector $2$ and for a fixed $\Omega$. The global maximum occurs for the linear ABA configuration, where neighboring detectors are separated by the minimal causal distance $L$ (Fig.~\ref{fig:ABA_arrangement}). In addition, a non-trivial local maximum occurs in the linear AAB configuration, where the two $A$ detectors are slightly offset from one another (Fig.~\ref{fig:AAB_arrangement}). 

For small energy gaps, where harvesting is strongest, nonzero harvesting occurs only near the ABA and AAB arrangements. In contrast, for larger energy gaps---where the two-detector harvesting condition $|X|>P$ is satisfied---harvesting occurs for all spatial positions of the third detector, as seen in Fig.~\ref{fig:arrangements3detectors} for the optimal two-detector energy gap.

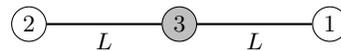
\begin{figure}[H]
    \centering
    \begin{tikzpicture}
        % Define nodes at specific coordinates
        \node (3) at (0,0) [draw, circle, fill=gray!50, inner sep=2pt] {3};
        \node (1) at (2,0) [draw, circle, fill=white, inner sep=2pt] {1};
        \node (2) at (-2,0) [draw, circle, fill=white, inner sep=2pt] {2};
        
        % Draw connecting lines
        \draw[thick] (2) -- (3) -- (1);     
    
        % Add side labels
        \node[below] at ($(3)!0.5!(1)$) {\( L \)};
        \node[below] at ($(3)!0.5!(2)$) {\( L \)};
    \end{tikzpicture}
    \caption{Three-detector ABA arrangement, with the detectors separated by the minimal causal distance $L$ corresponding to the globally optimal configuration at $\Omega T = 21.31$.}
    \label{fig:ABA_arrangement}
\end{figure}

\begin{figure}[H]
    \centering
    \begin{tikzpicture}
        % Define nodes at specific coordinates
        \node (3) at (0,0) [draw, circle, fill=gray!50, inner sep=2pt] {3};
        \node (1) at (2,0) [draw, circle, fill=white, inner sep=2pt] {1};
        \node (2) at (3,0) [draw, circle, fill=white, inner sep=2pt] {2};
        
        % Draw connecting lines
        \draw[thick] (3) -- (1);  
        \draw[thick, decorate, decoration={snake, amplitude=0.5mm, segment length=2mm},color=gray] (1) -- (2);
    
        % Add side labels
        \node[below] at ($(3)!0.5!(1)$) {\( L \)};
        \node[below] at ($(1)!0.5!(2)$) {\( x \)};
    \end{tikzpicture}
    \caption{Three-detector AAB arrangement corresponding to a locally optimal configuration, with one detector offset at $x/L=0.115$ and energy gap $\Omega T = 20.60$.}
    %$x/L=0.115499$
    \label{fig:AAB_arrangement}
\end{figure}
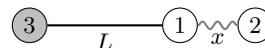

The corresponding locally maximal negativities, obtained by optimizing both the spatial configuration and the energy gap are
\begin{align}%24.49376
    \mathcal{N}^{(2)}_\text{ABA}(\Omega T = 21.31) &= 5.437 \cdot10^{-8}, \nonumber\\
    \mathcal{N}^{(2)}_\text{AAB}(\Omega T = 20.60) &= 1.650 \cdot10^{-8}, \\
    \mathcal{N}^{(2)}_\text{AB}(\Omega T = 24.49) &= 9.284 \cdot10^{-11}, \nonumber
\end{align}
where the last value, corresponding to the optimal two-detector configuration AB, is included for comparison with the three-detector cases.

\subsection{Four detectors}\label{sec:four-detect}
In this section, we explore entanglement harvesting in configurations involving four Unruh-DeWitt detectors. For the first time, we present an extensive analysis of how the spatial geometry of a four-detector setup influences entanglement harvesting. Introducing a fourth detector adds three spatial degrees of freedom to the optimization problem of maximizing negativity. %of identifying the optimal detector arrangement.
To avoid a high-dimensional multiparameter optimization, we restrict our analysis to a set of symmetric configurations, thereby reducing the problem to a single spatial parameter together with the detector energy gap. Within this framework, we derive analytic expressions for the leading-order negativity and numerically evaluate its behavior across the relevant parameter space.

For a system of four detectors, there are two possible bipartitions: a symmetric $2+2$ split and an asymmetric $3+1$ split. We focus primarily on the symmetric partition and consider a single asymmetric configuration for comparison.

\subsubsection{Symmetric $2+2$ partition}

For the $2+2$ partition, we get the $\tilde{\rho}_1$ block according to Eq.~\eqref{rho1_pt}
\begin{align}
    \tilde{\rho}_1 =
    \begin{pmatrix}
        \mathcal{C}_{BB}^* & \mathcal{X}^\dag_{BA} \\
        \mathcal{X}_{BA} & \mathcal{C}_{AA}
    \end{pmatrix} 
    &=
    \begin{pmatrix}
        P_{4} & C_{43} & X^*_{42} & X^*_{41} \\ 
        C^*_{43} & P_{3} & X^*_{32} & X^*_{31} \\ 
        X_{42} & X_{32} & P_{2} & C^*_{21} \\ 
        X_{41} & X_{31} & C_{21} & P_{1}
    \end{pmatrix} , \label{4detectors}
\end{align}
where, for identical detectors, we set $P_4=P_3=P_2=P_1\equiv P$. As in the three-detector case, the eigenvalue problem can be simplified by shifting $\alpha = y + P$, which eliminates $P$ from the characteristic equation and reduces it to a depressed quartic polynomial. Although the general roots can be expressed using Ferrari’s method, the resulting expressions are too cumbersome to write explicitly. However, for specific symmetric configurations, the negativity can still be expressed in a simple and intuitive form.

We analyze six different symmetric spatial configurations shown in Fig.~\ref{4detector_grid}. In these cases, with $X_{ij}\approx X_{ij}^+$, the components of $\tilde{\rho}_1$ in Eq.~\eqref{4detectors} are real-valued. When indexed by detector separation, many entries become equal, and the negativities can then be written explicitly as:
\begin{widetext}
    \begin{align}
        \mathcal{N}_\text{AABB} &= \sum_{s=\pm 1}\sum_{r=\pm 1} \max\Big\{0,\,-P + \tfrac{s}{2}(X^+_L+X^+_{L+2x}) + \tfrac{r}{2} \sqrt{4(C_x-sX^+_{L+x})^2+(X^+_{L+2x}-X^+_L)^2} \Big\} + \mathcal{O}(\lambda^4) , \label{AABB_neg}\\
        \mathcal{N}_\text{ABBA} &= \sum_{s=\pm 1}\sum_{r=\pm 1} \max\Big\{0,\,-P +\tfrac{s}{2}(C_x+C_{2L+x}) + \tfrac{r}{2} \sqrt{4(X^+_L-s X^+_{L+x})^2+(C_{x}-C_{2L+x})^2} \Big\} + \mathcal{O}(\lambda^4),\\
        \mathcal{N}_\text{ABAB} &= \sum_{s=\pm 1}\sum_{r=\pm 1}\max\Big\{0,\,-P +\tfrac{s}{2}(X^+_x+X^+_{3x}) + \tfrac{r}{2} \sqrt{4(C_{2x}+ X^+_{x})^2+(X^+_{x}-X^+_{3x})^2} \Big\} + \mathcal{O}(\lambda^4),\\
        \mathcal{N}_\text{rectangle} &= \sum_{s=\pm 1}\sum_{r=\pm 1}\max\Big\{0,\, -P +sC_{x}+rX^+_L+sr X^+_{\sqrt{L^2+x^2}}\Big\} + \mathcal{O}(\lambda^4),\label{neg_rectangle}\\
        \mathcal{N}_\text{skewedSquare} &= \sum_{r=\pm 1} \max\Big\{0,\, -P-\tfrac12(C_x+C_{\sqrt{4L^2-x^2}})+ \tfrac{r}{2} \sqrt{16 X^{+2}_L + \tfrac14(C_x-C_{\sqrt{4L^2-x^2}})^2} \Big\} + \mathcal{O}(\lambda^4),\\
        \mathcal{N}_\text{modTetrahedron} &= \sum_{r=\pm 1} \max\Big\{0,\, -P-C_x+2r X^+_L \Big\}+ \mathcal{O}(\lambda^4),\label{neg_modTetrahedron}
    \end{align}
\end{widetext}
with the last two negativity expressions omitting trivially non-negative eigenvalues of the form $P - C_{ij}$. The sum over $r$ and $s$ correspond to different eigenvalues possibly contributing to negativity.

The leading-order negativity expressions are evaluated over the spatial parameter $x$ and the energy gap $\Omega$ across their respective ranges, while $L$ is kept fixed at the minimal causal separation $L \equiv 2T$. The results for each configuration are presented in Fig.~\ref{4detector_grid} as contour plots of the leading-order negativity. In the right-hand column, we highlight the local maxima along with the corresponding parameter values at which they occur.

In the linear configurations AABB and ABBA, there are global maxima with slight offsets between detectors belonging to the same subsystem. This is consistent with the behavior observed in the three-detector AAB configuration shown in Fig.~\ref{fig:AAB_arrangement}.

Remarkably, the rectangle configuration is largely independent of the spatial parameter $x$, which appears as vertical lines in the contour plot in the Fig.~\ref{4detector_grid}. This independence becomes clear when analyzing the explicit eigenvalues of the leading-order negativity in Eq.~\eqref{neg_rectangle}. At low energy gaps, only one eigenvalue is negative, but as the gap increases, a second eigenvalue with opposite signs between $C_x$ and $X^+_{\sqrt{L^2+x^2}}$ becomes active, leaving the leading-order negativity
\begin{align*}
    \mathcal{N}_{\rm rectangle}\simeq -2(P+X^+_L) + \mathcal{O}(\lambda^4),
\end{align*}
which is independent of $x$. Notably, this value is exactly twice the leading-order negativity of the two-detector case in Eq.~\eqref{neg_2detectors}. This result is expected, since in the limit $x\to \infty$, the two detector pairs effectively reduce to independent subsystems, for which the general additivity of leading-order negativity in Eq.~\eqref{negativity_additivity} predicts precisely this behavior.

\begin{figure*}
    \centering
    \begin{tabular}{ |c|c|c| }

        \hline
        \textbf{Spatial configuration} & \textbf{Leading-order negativity} & \textbf{Local maxima} \\
        \hline

        \adjustbox{valign=c}{%
        \begin{subfigure}{0.25\textwidth}
        \scalebox{0.8}{\begin{tikzpicture}
            % Define nodes at specific coordinates
            \node (1) at (-2,0) [draw, circle, fill=white, inner sep=2pt] {1};
            \node (2) at (-1,0) [draw, circle, fill=white, inner sep=2pt] {2};
            \node (3) at (0.5,0) [draw, circle, fill=gray!50, inner sep=2pt] {3};
            \node (4) at (1.5,0) [draw, circle, fill=gray!50, inner sep=2pt] {4};
            
            % Draw connecting lines
            \draw[thick] (2) -- (3);   
            \draw[thick, decorate, decoration={snake, amplitude=0.5mm, segment length=2mm},color=gray] (1) -- (2);
            \draw[thick, decorate, decoration={snake, amplitude=0.5mm, segment length=2mm},color=gray] (3) -- (4);    
        
            % Add side labels
            \node[below] at ($(1)!0.5!(2)$) {\( x \)};
            \node[below] at ($(2)!0.5!(3)$) {\( L \)};
            \node[below] at ($(3)!0.5!(4)$) {\( x \)};
        \end{tikzpicture}}
        \caption{AABB, where $x/L\in [0,\infty)$}
        \end{subfigure}}
        &
        \adjustbox{valign=c}{%
        \begin{subfigure}{0.33\textwidth}
            \includegraphics[width=\textwidth]{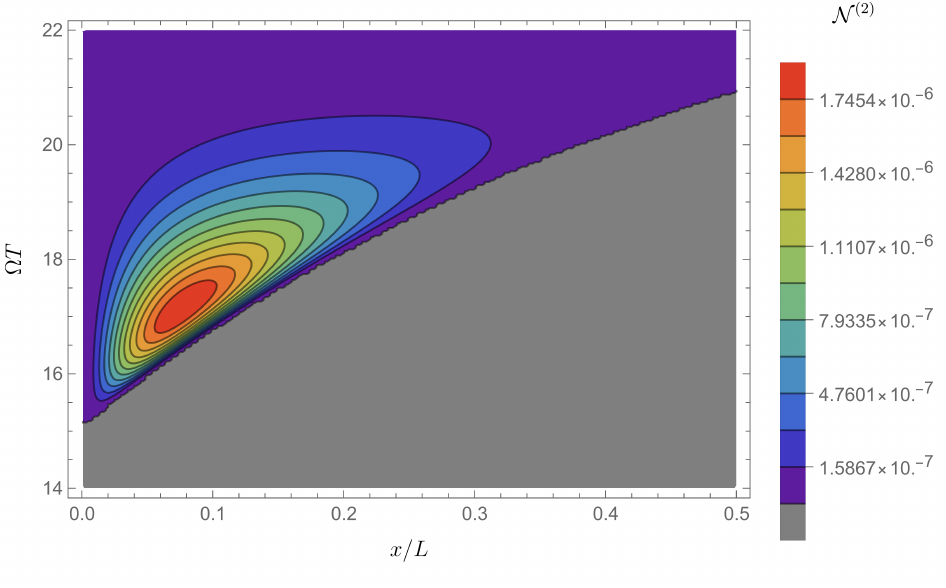}
        \end{subfigure}}
        &
        \hspace*{-0.6cm}
        \adjustbox{valign=c}{%
        \begin{subfigure}{0.25\textwidth}
        \begin{itemize}
            \item $\mathcal{N}^{(2)} = 1.904\cdot 10^{-6}$ with:
            \[
            \begin{cases}
                x/L = 0.077, \\
                \Omega T = 17.14.
            \end{cases}
            \]
        \end{itemize}
        \end{subfigure}}
        
        \\ \hline

        \adjustbox{valign=c}{%
        \begin{subfigure}{0.25\textwidth}
        \scalebox{0.8}{\begin{tikzpicture}
            % Define nodes at specific coordinates
            \node (1) at (-3,0) [draw, circle, fill=white, inner sep=2pt] {1};
            \node (4) at (-1.5,0) [draw, circle, fill=gray!50, inner sep=2pt] {4};
            \node (3) at (-0.5,0) [draw, circle, fill=gray!50, inner sep=2pt] {3};
            \node (2) at (1,0) [draw, circle, fill=white, inner sep=2pt] {2};
            
            % Draw connecting lines
            \draw[thick] (1) -- (4);   
            \draw[thick] (2) -- (3);   
            \draw[thick, decorate, decoration={snake, amplitude=0.5mm, segment length=2mm},color=gray] (3) -- (4);
        
            % Add side labels
            \node[below] at ($(1)!0.5!(4)$) {\( L \)};
            \node[below] at ($(4)!0.5!(3)$) {\( x \)};
            \node[below] at ($(3)!0.5!(2)$) {\( L \)};
        \end{tikzpicture}}
            \caption{ABBA, where $x/L\in [0,\infty)$}
        \end{subfigure}}
        &
        \adjustbox{valign=c}{%
        \begin{subfigure}{0.33\textwidth}
            \includegraphics[width=\textwidth]{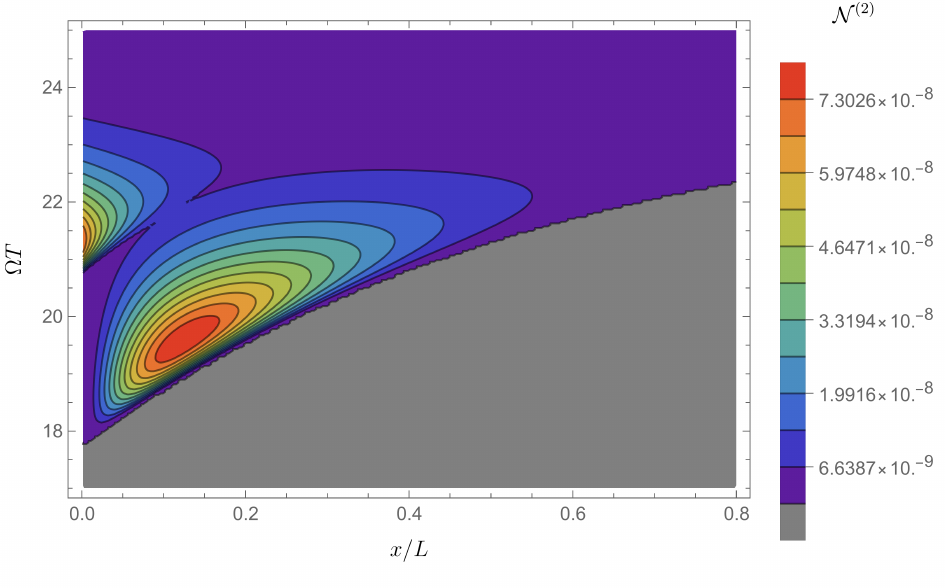}
        \end{subfigure}}
        &
        \hspace*{-0.6cm}
        \adjustbox{valign=c}{%
        \begin{subfigure}{0.25\textwidth}
        \begin{itemize}
            \item $\mathcal{N}^{(2)} = 7.970 \cdot10^{-8}$ with:
            \[
            \begin{cases}
                x/L = 0.124, \\
                \Omega T = 19.58,
            \end{cases}
            \]
            \item $\mathcal{N}^{(2)} = 7.451 \cdot 10^{-8}$ with:
              \[
                \begin{cases}
                  x/L = 0, \\
                  \Omega T = 21.31 .
                \end{cases}
              \]
        \end{itemize}
        \end{subfigure}}
        
        \\ \hline

        \adjustbox{valign=c}{%
        \begin{subfigure}{0.25\textwidth}
        \scalebox{0.8}{\begin{tikzpicture}
            % Define nodes at specific coordinates
            \node (1) at (-1.5,0) [draw, circle, fill=white, inner sep=2pt] {1};
            \node (3) at (0,0) [draw, circle, fill=gray!50, inner sep=2pt] {3};
            \node (2) at (1.5,0) [draw, circle, fill=white, inner sep=2pt] {2};
            \node (4) at (3,0) [draw, circle, fill=gray!50, inner sep=2pt] {4};
            
            % Draw connecting lines (only one)
            \draw[thick, decorate, decoration={snake, amplitude=0.5mm, segment length=2mm},color=gray] (1) -- (4);   

            % Draw again to put in front (piece-wise drawn lines did not work nicely)
            \node (1) at (-1.5,0) [draw, circle, fill=white, inner sep=2pt] {1};
            \node (3) at (0,0) [draw, circle, fill=gray!50, inner sep=2pt] {3};
            \node (2) at (1.5,0) [draw, circle, fill=white, inner sep=2pt] {2};
            \node (4) at (3,0) [draw, circle, fill=gray!50, inner sep=2pt] {4};
        
            % Add side labels
            \node[below] at ($(1)!0.5!(3)$) {\( x \)};
            \node[below] at ($(3)!0.5!(2)$) {\( x \)};
            \node[below] at ($(2)!0.5!(4)$) {\( x \)};
        \end{tikzpicture}}
            \caption{ABAB, where $x/L\in [1,\infty)$}
        \end{subfigure}}
        &
        \adjustbox{valign=c}{%
        \begin{subfigure}{0.33\textwidth}
            \includegraphics[width=\textwidth]{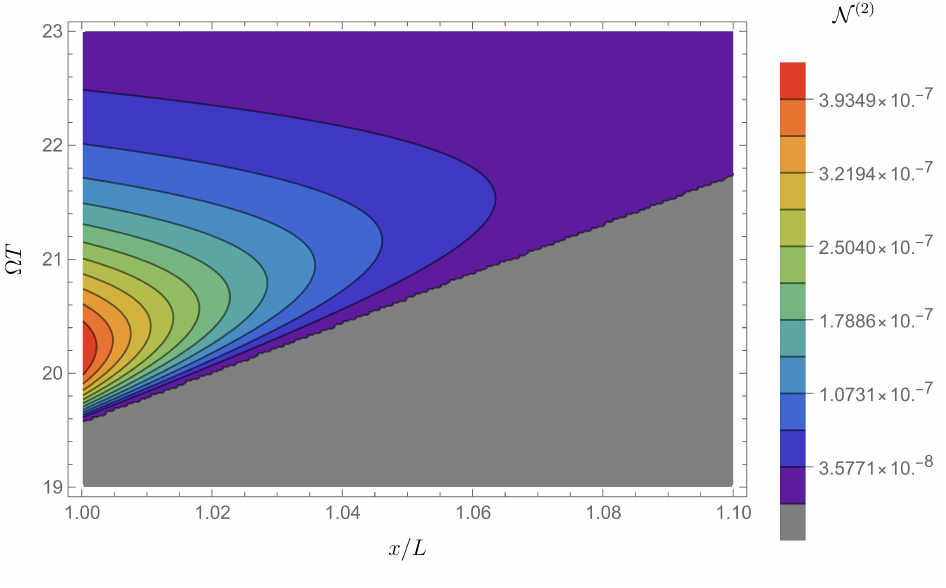}
        \end{subfigure}}
        &
        \hspace*{-0.6cm}
        \adjustbox{valign=c}{%
        \begin{subfigure}{0.25\textwidth}
        \begin{itemize}
            \item $\mathcal{N}^{(2)} =  4.293 \cdot 10^{-7}$ with:
            \[
            \begin{cases}
                x/L = 1, \\
                \Omega T = 20.19.
            \end{cases}
            \]
        \end{itemize}
        \end{subfigure}}
        
        \\ \hline

        \adjustbox{valign=c}{%
        \begin{subfigure}{0.25\textwidth}
        \scalebox{0.8}{\begin{tikzpicture}
            % Define nodes at specific coordinates
            \node (1) at (0,0) [draw, circle, fill=white, inner sep=2pt] {1};
            \node (2) at (0,2) [draw, circle, fill=white, inner sep=2pt] {2};
            \node (3) at (3,2) [draw, circle, fill=gray!50, inner sep=2pt] {3};
            \node (4) at (3,0) [draw, circle, fill=gray!50, inner sep=2pt] {4};
            
            % Draw connecting lines
            \draw[thick] (2) -- (3);   
            \draw[thick] (4) -- (1);   
            \draw[thick, decorate, decoration={snake, amplitude=0.5mm, segment length=2mm},color=gray] (1) -- (2);
            \draw[thick, decorate, decoration={snake, amplitude=0.5mm, segment length=2mm},color=gray] (3) -- (4);
        
            % Add side labels
            \node[left] at ($(1)!0.5!(2)$) {\( x \)};
            \node[right] at ($(3)!0.5!(4)$) {\( x \)};
            \node[above] at ($(2)!0.5!(3)$) {\( L \)};
            \node[below] at ($(1)!0.5!(4)$) {\( L \)};
        \end{tikzpicture}}
            \caption{Rectangle, where $x/L\in [0,\infty)$}
        \end{subfigure}}
        &
        \adjustbox{valign=c}{%
        \begin{subfigure}{0.33\textwidth}
            \includegraphics[width=\textwidth]{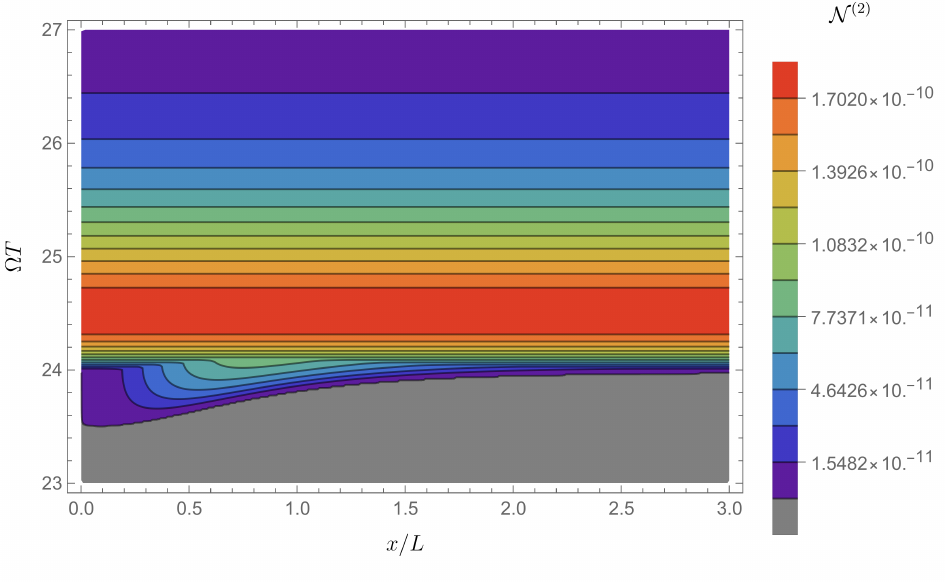}
        \end{subfigure}}
        &
        \hspace*{-0.6cm}
        \adjustbox{valign=c}{%
        \begin{subfigure}{0.25\textwidth}
        \begin{itemize}
            \item $\mathcal{N}^{(2)} = 1.857 \cdot 10^{-10}$ with:
            \[
            \begin{cases}
                x/L\in [0,\infty), \\
                \Omega T = 24.49 .
            \end{cases}
            \]
            The maximum $\mathcal{N}^{(2)}$ is two times the maximum $\mathcal{N}^{(2)}$ of two detector setup.
        \end{itemize}
        \end{subfigure}}
        
        \\ \hline

        \adjustbox{valign=c}{%
        \begin{subfigure}{0.25\textwidth}
        \scalebox{0.8}{\begin{tikzpicture}
            % Define nodes at specific coordinates
            \node (1) at (0,-1) [draw, circle, fill=white, inner sep=2pt] {1};
            \node (2) at (0,1) [draw, circle, fill=white, inner sep=2pt] {2};
            \node (3) at (2,0) [draw, circle, fill=gray!50, inner sep=2pt] {3};
            \node (4) at (-2,0) [draw, circle, fill=gray!50, inner sep=2pt] {4};
            
            % Draw connecting lines
            \draw[thick] (1) -- (3) -- (2) -- (4) -- (1);  
            \draw[thick, decorate, decoration={snake, amplitude=0.5mm, segment length=2mm},color=gray] (1) -- (2);
        
            % Add side labels
            \node[right] at ($(1)!0.5!(2)$) {\( x \)};
            \node[below right] at ($(1)!0.5!(3)$) {\( L \)};
            \node[below left] at ($(1)!0.5!(4)$) {\( L \)};
            \node[above right] at ($(3)!0.5!(2)$) {\( L \)};
            \node[above left] at ($(2)!0.5!(4)$) {\( L \)};
        \end{tikzpicture}}
            \caption{Skewed square, where $x/L\in [0,\sqrt2]$}
        \end{subfigure}}
        &
        \adjustbox{valign=c}{%
        \begin{subfigure}{0.33\textwidth}
            \includegraphics[width=\textwidth]{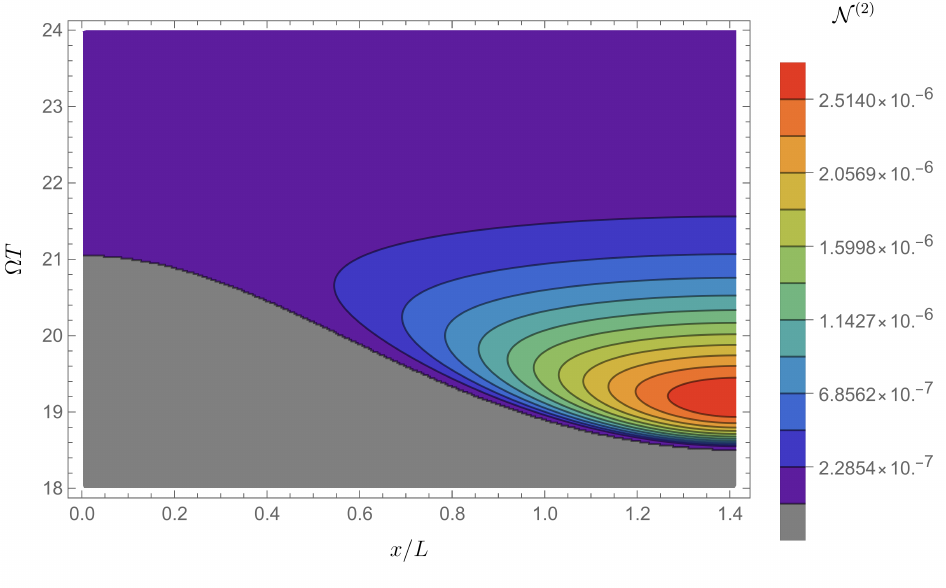}
        \end{subfigure}}
        &
        \hspace*{-0.6cm}
        \adjustbox{valign=c}{%
        \begin{subfigure}{0.25\textwidth}
        \begin{itemize}
            \item $\mathcal{N}^{(2)} = 2.743 \cdot 10^{-6}$ with:
            \[
            \begin{cases}
                x/L = \sqrt2 , \\
                \Omega T = 19.16 .
            \end{cases}
            \]
        \end{itemize}
        \end{subfigure}}
        
        \\ \hline 

        \adjustbox{valign=c}{%
        \begin{subfigure}{0.25\textwidth}
        \tdplotsetmaincoords{60}{120}
        \scalebox{0.8}{\begin{tikzpicture}[tdplot_main_coords]
            % Define vertices
            \coordinate (1) at (0, 0, 0);
            \coordinate (2) at (0, 1, 2);
            \coordinate (3) at (1, 3, 0);
            \coordinate (4) at (-1, 3, 0);
    
            % Draw edges and label side lengths
            \draw[thick, decorate, decoration={snake, amplitude=0.5mm, segment length=3mm},color=gray] (1) -- (2) node[midway, left] {\color{black}{\small $x$}};
            \draw[thick] (1) -- (3) node[midway, below left] {\small $L$};
            \draw[dashed] (1) -- (4) node[midway, above right] {\small $L$};
            \draw[thick] (2) -- (3) node[midway, left] {\small $L$};
            \draw[thick] (2) -- (4) node[midway, above right] {\small $L$};
            \draw[thick, decorate, decoration={snake, amplitude=0.5mm, segment length=2mm},color=gray] (3) -- (4) node[midway, below right] {\color{black}{\small $x$}};

            % Draw vertex nodes as balls with labels inside
            \node[shade, circle, draw, minimum size=18pt, inner sep=0pt,
          top color=white, bottom color=gray!20] at (1) {1};
            \node[shade, circle, draw, minimum size=18pt, inner sep=0pt,
          top color=white, bottom color=gray!20] at (2) {2};
            \node[shade, ball color=gray!50, circle, draw, minimum size=18pt, inner sep=0pt] at (3) {3};
            \node[shade, ball color=gray!50, circle, draw, minimum size=18pt, inner sep=0pt] at (4) {4};
        
        \end{tikzpicture}}
            \caption{Modified tetrahedron, where $x/L\in [0,\sqrt2]$}
        \end{subfigure}}
        &
        \adjustbox{valign=c}{%
        \begin{subfigure}{0.35\textwidth}
            \includegraphics[width=\textwidth]{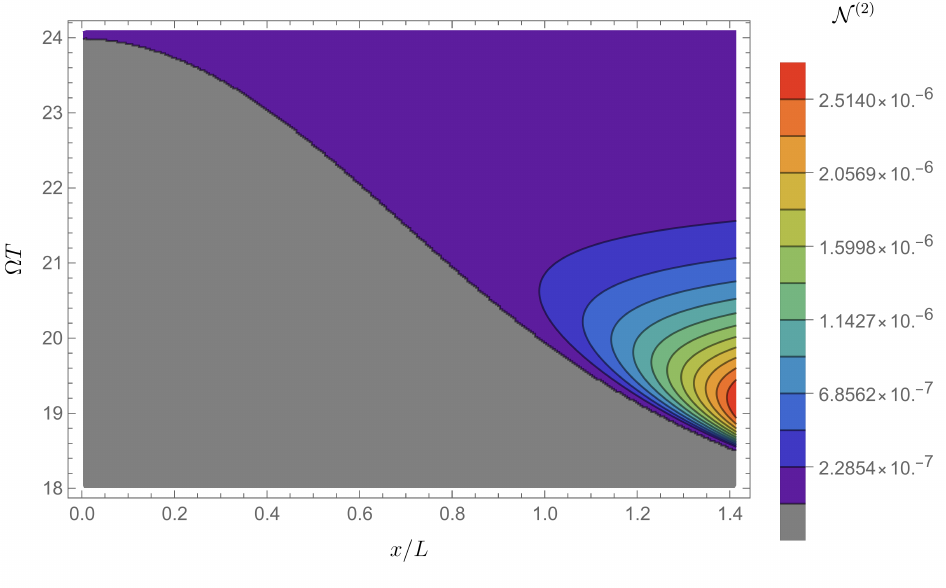}
            %\caption{Plot 2}
        \end{subfigure}}
        &
        \hspace*{-0.6cm}
        \adjustbox{valign=c}{%
        \begin{subfigure}{0.25\textwidth}
        \begin{itemize}
            \item $\mathcal{N}^{(2)} = 2.743 \cdot 10^{-6}$ with:
            \[
            \begin{cases}
                x/L = \sqrt2 , \\
                \Omega T = 19.16 .
            \end{cases}
            \]
        \end{itemize}
        \end{subfigure}}
        \\ \hline
    \end{tabular}
    \caption{Contour plots of the leading-order negativity for each of the symmetric four-detector configurations considered. The negativity is evaluated as a function of the spatial parameter $x$ and the detector energy gap $\Omega$, with $L$ fixed at the minimal causal separation ($L\equiv 2T=10\sigma$). The grey region corresponds to zero negativity. Local maxima and their corresponding parameter values are indicated.}
    \label{4detector_grid}
\end{figure*}
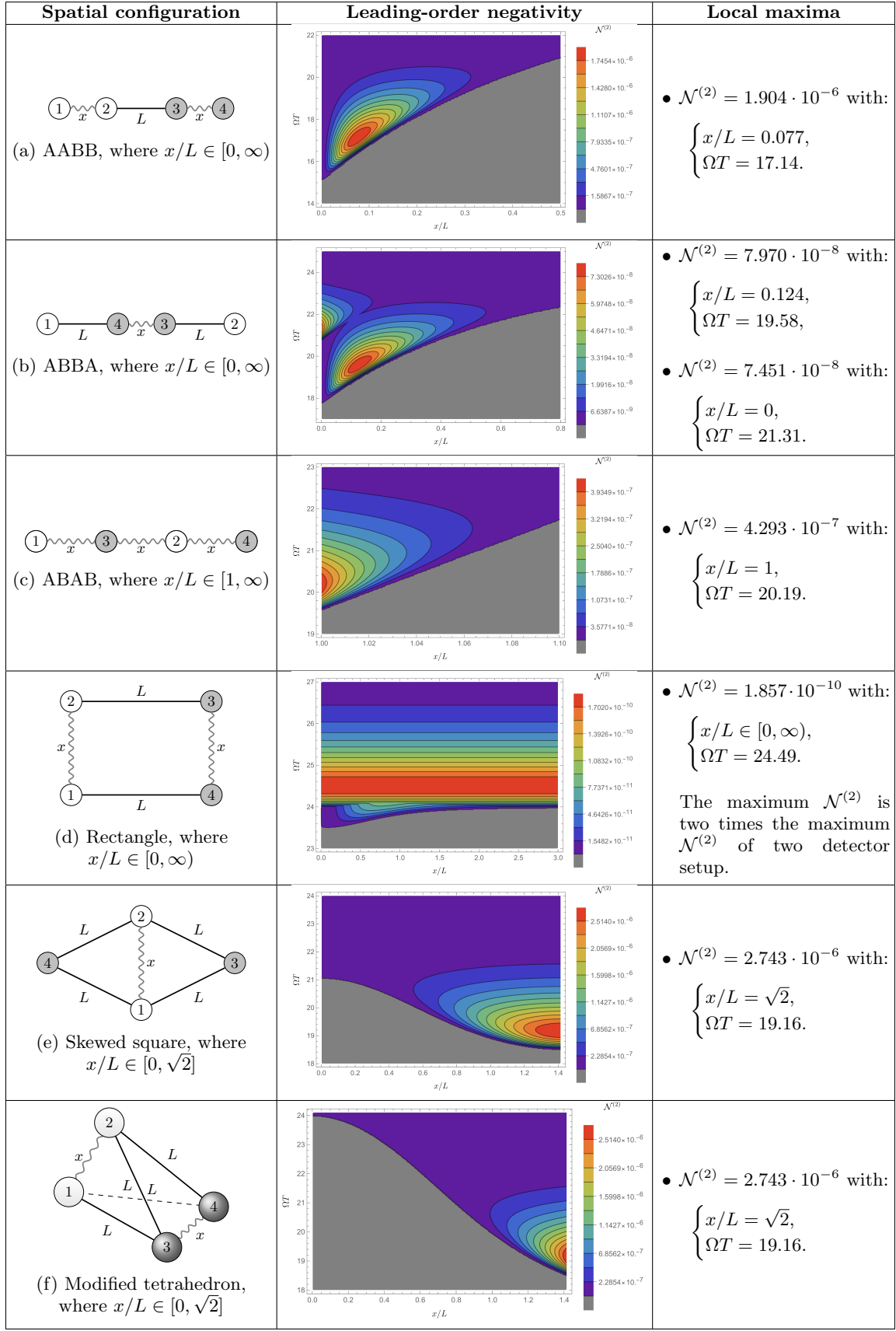

Among the configurations studied, the optimal four-detector arrangement is the so-called diagonal square configuration shown in Fig.~\ref{fig:diagonal_square_arr}. This configuration appears as the global maximum for both the skewed square and modified tetrahedron setups shown in Fig.~\ref{4detector_grid}, occurring at $x/L=\sqrt{2}$. Maximal entanglement is obtained when detectors belonging to different subsystems are placed in close proximity, while detectors within the same subsystem are maximally separated.

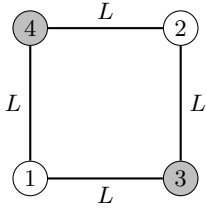
\begin{figure}[htb]
    \centering
        \begin{tikzpicture}
            % Define nodes at specific coordinates
            \node (1) at (0,0) [draw, circle, fill=white, inner sep=2pt] {1};
            \node (2) at (0,2) [draw, circle, fill=gray!50, inner sep=2pt] {4};
            \node (3) at (2,2) [draw, circle, fill=white, inner sep=2pt] {2};
            \node (4) at (2,0) [draw, circle, fill=gray!50, inner sep=2pt] {3};
            
            % Draw connecting lines
            \draw[thick] (2) -- (3);   
            \draw[thick] (4) -- (1);   
            \draw[thick] (1) -- (2);
            \draw[thick] (3) -- (4);
        
            % Add side labels
            \node[left] at ($(1)!0.5!(2)$) {\( L \)};
            \node[right] at ($(3)!0.5!(4)$) {\( L \)};
            \node[above] at ($(2)!0.5!(3)$) {\( L \)};
            \node[below] at ($(1)!0.5!(4)$) {\( L \)};
        \end{tikzpicture}
    \caption{The optimal four-detector configuration is so-called diagonal square arrangement where square side length is the minimal causal distance $L$ and detector energy gap is $\Omega T = 19.16$.}
    \label{fig:diagonal_square_arr}
\end{figure}

Finally, we emphasize that the number of detectors has a significant impact on the magnitude of the leading-order negativity. This is evident when comparing the optimal configurations for two detectors, three detectors (ABA in Fig.~\ref{fig:ABA_arrangement}), and four detectors (diagonal square in Fig.~\ref{fig:diagonal_square_arr}). As the number of detectors increases, entanglement can be harvested over a substantially broader range of energy gaps $\Omega$, leading to larger negativities that differ by several orders of magnitude:
\begin{align}
    \mathcal{N}^{(2)}_{\rm 4-det}(\Omega T = 19.16) &= 2.743 \cdot10^{-6}, \nonumber\\
    \mathcal{N}^{(2)}_{\rm 3-det}(\Omega T = 21.31) &= 5.437 \cdot10^{-8}, \\
    \mathcal{N}^{(2)}_{\rm 2-det}(\Omega T = 24.49) &= 9.284 \cdot10^{-11}. \nonumber
\end{align}
In Sec.~\ref{sec:Comparison with compactly supported switching function}, we repeat the analysis using different compactly supported switching functions. The optimal arrangements are unchanged and switching-independent when adding third and fourth detectors, while the widening of the harvesting regime with respect to the energy gap---and the scaling of the leading-order negativity with additional detectors---depends on the chosen switching.

\subsubsection{Asymmetric $3+1$ partition}

For the $3+1$ partition, we get the $\tilde{\rho}_1$ block according to Eq.~\eqref{rho1_pt}:
\begin{align}
    \tilde{\rho}_1 =
    \begin{pmatrix}
        \mathcal{C}_{BB}^* & \mathcal{X}^\dag_{BA} \\
        \mathcal{X}_{BA} & \mathcal{C}_{AA}
    \end{pmatrix} 
    &=
    \begin{pmatrix}
        P_{4} & X^*_{43} & X^*_{42} & X^*_{41} \\ 
        X_{43} & P_{3} & C^*_{32} & C^*_{31} \\ 
        X_{42} & C_{32} & P_{2} & C^*_{21} \\ 
        X_{41} & C_{31} & C_{21} & P_{1}
    \end{pmatrix}.\label{4detectors_asymm}
\end{align}
For an asymmetric partition, we consider a setup with two degrees of freedom in the spatial configuration. The setup is shown in Fig.~\ref{fig:asymm}, where we fix the minimal distance $L$ between the detectors of different subsystems but vary the angles as
\begin{align*}
    \theta_{21}&\in [0,2\pi),  \\
    \theta_{32}&\in [0,2\pi-\theta_{21}).
\end{align*}
Using the distance-based indexing for the $\mathcal{X}_{BA}$ entries, the block in Eq.~\eqref{4detectors_asymm} simplifies to:
\begin{align}
    \tilde{\rho}_1 &=
    \begin{pmatrix}
        P & X^*_L & X^*_L & X^*_L \\ 
        X_L & P & C^*_{32} & C^*_{31} \\ 
        X_L & C_{32} & P & C^*_{21} \\ 
        X_L & C_{31} & C_{21} & P
    \end{pmatrix}, \label{4detectors_asymm_special}
\end{align}
where the $C_{ij}$ terms depend on the detector separations which again depend on the angles varied:
\begin{align*}
    x_{21} &= 2L\sin\big(\tfrac{\theta_{21}}{2}\big), \\
    x_{32} &= 2L\sin\big(\tfrac{\theta_{32}}{2}\big), \\
    x_{31} &= 2L\sin\big(\tfrac{2\pi-\theta_{21}-\theta_{32}}{2}\big).
\end{align*}

By varying the angles $\theta_{21}$, $\theta_{32}$ and the energy gap $\Omega$, we find the optimal configuration as the equilateral triangle configuration with parameters:
\begin{align*}
    \theta_{21} &= \theta_{32}=2\pi/3, \\
    \Omega T &= 20.03,
\end{align*}
yielding leading-order negativity
\begin{align*}
    \mathcal{N}^{(2)}_{3+1}=5.723\cdot10^{-7}.
\end{align*}
These results are consistent with the pattern observed in the previous three-detector and symmetric four-detector analyses: optimal entanglement harvesting is achieved by minimizing the distance between detectors in different subsystems while maximizing the separation between detectors within the same subsystem.

\begin{figure*}[htbp]
    \centering
    \begin{minipage}[c]{0.48\linewidth}
        \centering
        \begin{tikzpicture}
            % Define coordinates
            \node (origin) at (0,0) [draw, circle, fill=gray!50, inner sep=1pt] {4};
            \node (D1) at (2,0) [draw, circle, fill=white, inner sep=1pt] {1};
            \node (D2) at (-1.593,1.210) [draw, circle, fill=white, inner sep=1pt] {2};
            \node (D3) at (-0.251,-1.984) [draw, circle, fill=white, inner sep=1pt] {3};
    
            % Draw connecting lines
            \draw[thick] (origin) -- (D1);
            \draw[thick] (origin) -- (D2);
            \draw[thick] (origin) -- (D3);
            \draw[thick, opacity=0.3] (D1) -- (D2);
            \draw[thick, opacity=0.3] (D2) -- (D3);
            \draw[thick, opacity=0.3] (D1) -- (D3);
    
            % Draw angles
            \pic [draw, angle radius=3mm, angle eccentricity=1.6, "$\theta_{21}$"] {angle=D1--origin--D2};
            \pic [draw, angle radius=3.8mm, angle eccentricity=1.6, "$\theta_{32}$"] {angle=D2--origin--D3};
    
            \draw[opacity=0.3] (origin) circle(2);    
    
            \node[below] at ($(origin)!0.5!(D1)$) {\( L \)};
            \node[above] at ($(D1)!0.5!(D2)$) {\( x_{21} \)};
            \node[left] at ($(D2)!0.5!(D3)$) {\( x_{32} \)};
            \node[right] at ($(D1)!0.5!(D3)$) {\( x_{31} \)};
        \end{tikzpicture}
        \vspace{2cm}
        \label{fig:arrangements3}
    \end{minipage}
    \begin{minipage}[c]{0.48\linewidth}
        \centering
        \includegraphics[width=0.85\textwidth]{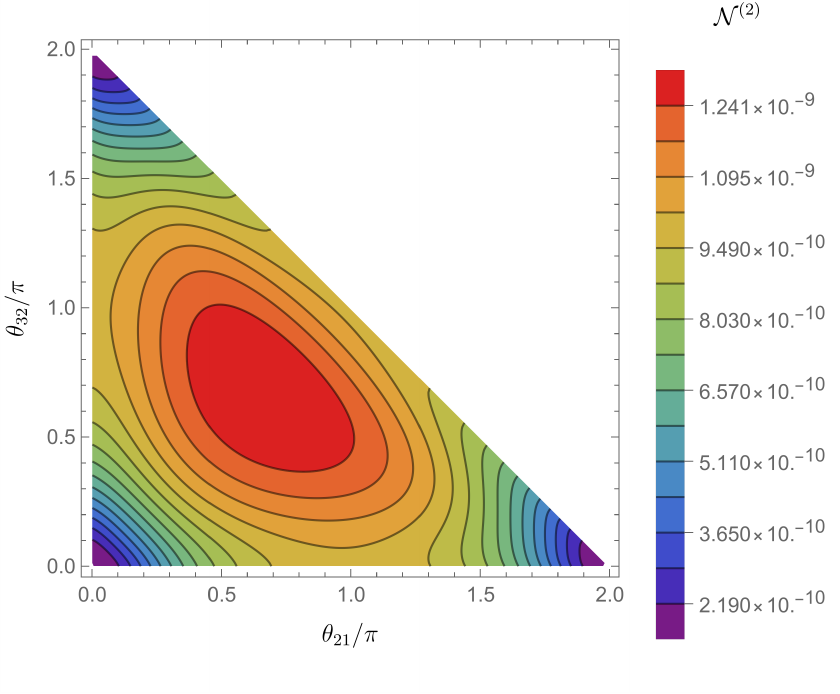}
        \label{fig:3detectors-density}
    \end{minipage}
\caption{
On the left is an illustration of an asymmetric $3+1$ partition configuration. By varying the angles $\theta_{21}$ and $\theta_{32}$, we obtain all planar geometric configurations with minimal separation $L\equiv 2T=10\sigma$ between different subsystem detectors. On the right is the corresponding density plot of the leading-order negativity, computed using the optimal two-detector energy gap ($\Omega T = 24.49$) as a representative example. In the optimal configuration subsystem $A$ detectors form an equilateral triangle around the center detector. 
}
\label{fig:asymm}
\end{figure*}
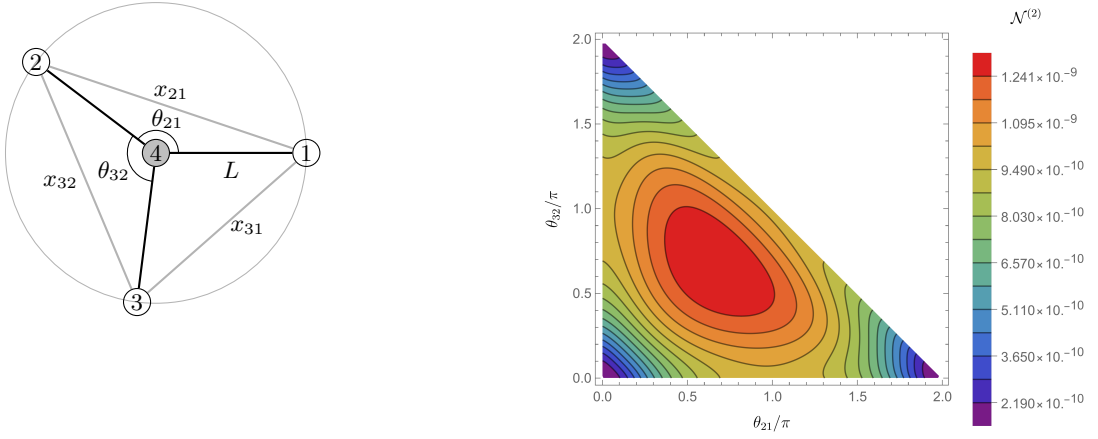

\subsubsection{Varying the spatial scale of the optimal arrangements}\label{Varying the spatial scale of the optimal configurations}

In this brief subsection we explore how entanglement changes when we rescale the distances between the detectors. In particular, we compute the leading-order negativity for different spatial arrangements as a function of the distance $l$ between subsystems $A$ and $B$. In Fig.~\ref{fig:scale-comparison} we compare the two-detector case with several optimal multi-detector configurations: the three-detector ABA (Fig.~\ref{fig:ABA_arrangement}), the four-detector diagonal square (Fig.~\ref{fig:diagonal_square_arr}), the asymmetric equilateral triangle, and the AABB configuration with $x/l = 0.08$ (Fig.~\ref{4detector_grid}).

From this analysis we see that for the same energy gap, adding more detectors allows to harvest entanglement for larger distances between subsystems. Among all the configurations explored, the four detectors AABB configuration is the one that allows entanglement harvesting for the largest distances between the $A$ and $B$ subsystems.

\begin{figure}[H]
    \centering
    \includegraphics[width=0.48\textwidth]{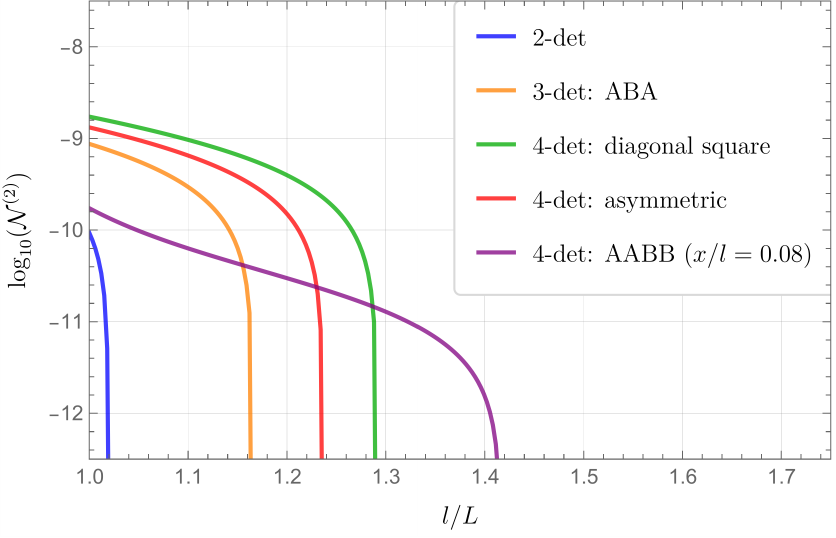}
    \caption{Leading-order negativity in $\log_{10}$ scale as a function of $l/L$, where $l$ is the distance between subsystems $A$ and $B$ and $L \equiv 2T$ is the minimal causal separation. The energy gap used to plot all the configurations is the optimal one for two detector $\Omega T=24.49$.}
    \label{fig:scale-comparison}
\end{figure}

\subsection{Alternative linear chain}\label{sec:Alternative linear chain}

In this section, we study the behavior of the leading-order negativity in a model where the detectors are arranged on a one-dimensional lattice with alternating subsystems $A$ and $B$. The detectors are separated by a uniform spacing $l$ between nearest neighbors, corresponding to the minimal causal separation. The setup is illustrated in Fig.~\ref{fig:linear_chain}.

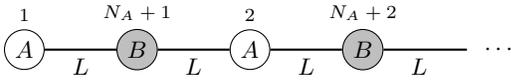
\begin{figure}
    \centering
    \begin{tikzpicture}
        % Define nodes at specific coordinates
        \node (1) at (0,0) [draw, circle, fill=white, minimum size=15pt, inner sep=2pt, label=above:{\scriptsize$1$}] {$A$};
        \node (2) at (1.5,0) [draw, circle, fill=gray!50, minimum size=15pt, inner sep=2pt, label=above:{\scriptsize$N_A+1$}] {$B$};
        \node (3) at (3,0) [draw, circle, fill=white, minimum size=15pt, inner sep=2pt, label=above:{\scriptsize$2$}] {$A$};
        \node (4) at (4.5,0) [draw, circle, fill=gray!50, minimum size=15pt, inner sep=2pt, label=above:{\scriptsize$N_A+2$}] {$B$};
        \node (5) at (6,0) {};
        % Draw connecting lines
        \draw[thick] (1) -- (2) -- (3) -- (4) -- (5);     
        % Add side labels
            \node[below] at ($(1)!0.5!(2)$) {\( L \)};
            \node[below] at ($(2)!0.5!(3)$) {\( L \)};
            \node[below] at ($(3)!0.5!(4)$) {\( L \)};
            \node[below] at ($(4)!0.5!(5)$) {\( L \)};
            \node[right] at ($(5)$) {\( \cdots \)};
    \end{tikzpicture}
    \caption{Linear chain setup with alternating subsystem detectors with minimal causal separation $L\equiv 2T$. Indices are shown above each detector.}
    \label{fig:linear_chain}
\end{figure}

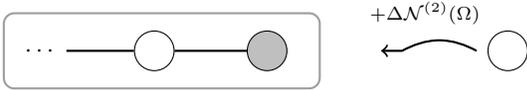
\begin{figure}
    \centering
    \begin{tikzpicture}
        % Draw the box first
        \draw[thick, rounded corners, color=gray!75] (-2,-0.5) rectangle (2.2,0.5);
    
        % Define nodes
        \node (0) at (-1.5,0) {$\cdots$};
        \node (1) at (0,0) [draw, circle, fill=white, minimum size=15pt] {};
        \node (2) at (1.5,0) [draw, circle, fill=gray!50, minimum size=15pt] {};
        \node (2b) at (2.9,0) {};

        \node (3) at (4.7,0) [draw, circle, fill=white, minimum size=15pt] {};
        \node (3b) at (4.4,0)  {};
        \node (3c) at (3.6,0.5)  {\scriptsize $+\Delta\mathcal{N}^{(2)}(\Omega)$};
        
        % Draw connecting lines
        \draw[thick] (0) -- (1) -- (2);
        \draw[decorate, decoration={coil, amplitude=1.4mm, segment length=17mm, mirror}, ->, thick] (3b) -- (2b);
    node[midway, above] {};
    \end{tikzpicture}
    \caption{Infinite-chain limit of Fig.~\ref{fig:linear_chain}. Each additional detector increases the leading-order negativity $\Delta\mathcal{N}^{(2)}(\Omega)$, and the optimal energy gap occurs at $\Omega T \simeq 18.88$.}
    \label{fig:linear_chain_limit}
\end{figure}

In order to compute the leading-order negativity, we construct the block $\tilde \rho_1$ from Eq.~\eqref{rho1_pt} for the present setup, where the indices of each subblock, as shown in Eq.~\eqref{rho1_block_opened}, follow the chain indexing in Fig.~\ref{fig:linear_chain}. By using the distance-dependent index labeling, $\mathcal{C}_{BB}$ and $\mathcal{C}_{AA}$ become Toeplitz matrices with entries $C_{2 n L}$ where the detector separation increases with the distance from the diagonal. Similarly, the off-diagonal block $\mathcal{X}_{BA}$ exhibits a Toeplitz structure with a banded pattern of $X_{(2n-1)L}$ terms. The pattern of these bands depends on whether the total number of detectors is even or odd. The explicit construction of $\tilde\rho_1$ for an arbitrary number of detectors is provided in Appendix~\ref{app:linear_chain_rho1}.

\begin{figure*}[htbp]
    \centering
    % Left large figure
    \begin{minipage}[c]{0.58\linewidth}
        \centering
        \includegraphics[width=\textwidth]{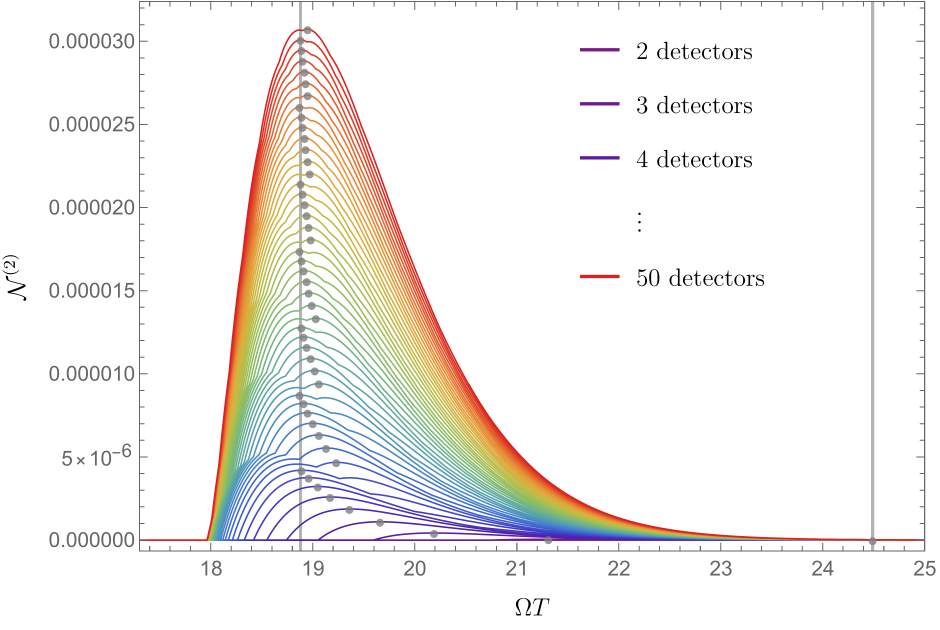}
    \end{minipage}%
    \hfill
    % Right two stacked smaller figures
    \begin{minipage}[c]{0.38\linewidth}
        \centering
        \includegraphics[width=\textwidth]{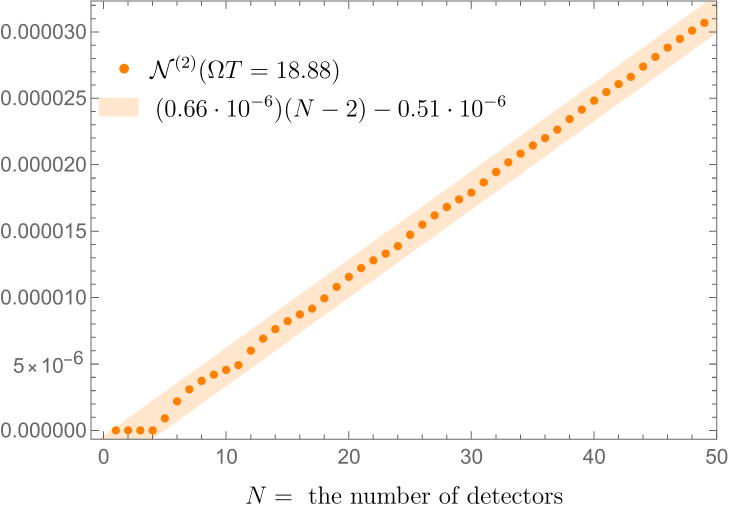}
        \vspace{0.5cm}
        \includegraphics[width=\textwidth]{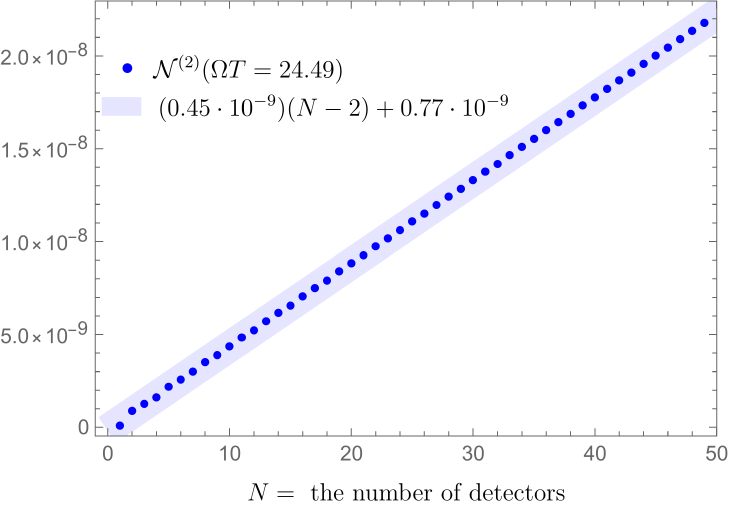}
    \end{minipage}
    
    \caption{On the left, we show the leading-order negativities for the alternative chain setup with up to 50 detectors as a function of the energy gaps $\Omega$. Global maxima are marked with gray dots, and the special cases --- the two-detector and infinite-detector limits --- with their optimal energy gap values indicated by vertical lines. In general, the negativity grows roughly linearly with the number of detectors for all $\Omega$, with the special-case values plotted on the right and fitted with a linear trend.}
    \label{fig:mountain}
\end{figure*}

Figure \ref{fig:mountain} (left panel) shows the results of the evaluated leading-order negativity as a function of the detector energy levels up to 50 detectors. From three detectors onward, the negativity grows roughly linearly. For a large number of detectors this is expected: in the infinite-chain limit, adding a single detector has little effect on the overall configuration and entanglement structure. Since negativity decays rapidly with distance between subsystems, each new detector primarily introduces entanglement with its nearest neighbors. The infinite-chain limit is illustrated in Fig.~\ref{fig:linear_chain_limit}.

As the number of detectors grows in the chain, the optimal energy gap converges to $\Omega T \simeq 18.88$ (see Fig.~\ref{fig:mountain}, left panel). In this large-detector regime, each additional detector contributes an approximately constant increase in negativity, $\Delta\mathcal{N}^{(2)}\simeq 0.66\cdot 10^{-6}$, as illustrated by the linear fit in the bottom-right panel. The negativity is computed from the negative eigenvalues of the partially transposed and linearly scaling block $\tilde \rho_1$; as more detectors are added, new eigenvalue contributions appear as small bumps in the left panel.

Finally, let us address the validity of the perturbative approach in this configuration. Although the leading-order negativity scales linearly with the number of detectors, $\mathcal{N}^{(2)} \propto N$, the perturbative expansion may break down for large $N$ at a fixed coupling $\lambda$. As \cite{Pozas-Kerstjens:2015, Morote-Balboa:2026dfe} suggest, beyond a certain chain length it may become necessary to include next-to-leading-order corrections or to adopt a fully non-perturbative treatment \cite{Simidzija-Martin2017, Polo-Martin2024}. This consideration is especially relevant for experimental implementations based on detector models, such as those proposed in \cite{Gooding:2023xxl, Perche:2025buo}, where $\lambda$ is constrained by the physical setup. In such cases, predicted negativities may already approach the regime in which leading-order perturbation theory becomes unreliable \cite{Morote-Balboa:2026dfe}.

However, determining how higher-order contributions to the negativity scale with $N$ is nontrivial. It requires identifying which eigenvalues of the partially transposed density matrix become negative at a given order, evaluating their magnitudes, and understanding how their sum scales with $N$. Each of these steps demands careful analysis, which we leave for future work.

\subsection{Comparison with compactly supported switching functions}\label{sec:Comparison with compactly supported switching function}

We performed the main analysis of section \ref{sec:Searching for the optimal spatial configurations} using pointlike detectors with Gaussian switching, with the interaction effectively confined to $T = 5\sigma$ as described in Sec.~\ref{sec:setup_and_switching}. For completeness, we now compare those results with switching functions that are strictly compactly supported, namely truncated Gaussians and compactified polynomials. Smooth compactly supported functions (e.g., bump functions) were also considered, but do not yield non-zero leading-order negativity in this setup. In all comparisons, we fix the total interaction strength $\int \chi(t)\,dt=1$ and restrict the support to the finite interval $t \in [-T,T]$.

\subsubsection{Truncated Gaussian switching}

\begin{figure*}[htbp]
    \centering
    \hspace*{0.05\textwidth}%
    \includegraphics[width=0.85\textwidth]{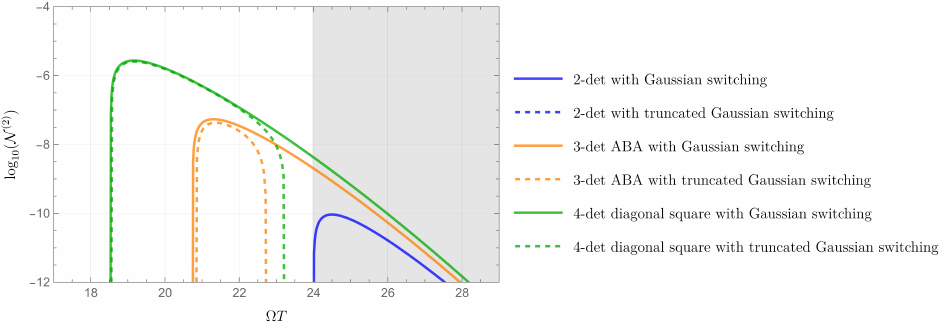}
    \includegraphics[width=0.74\textwidth]{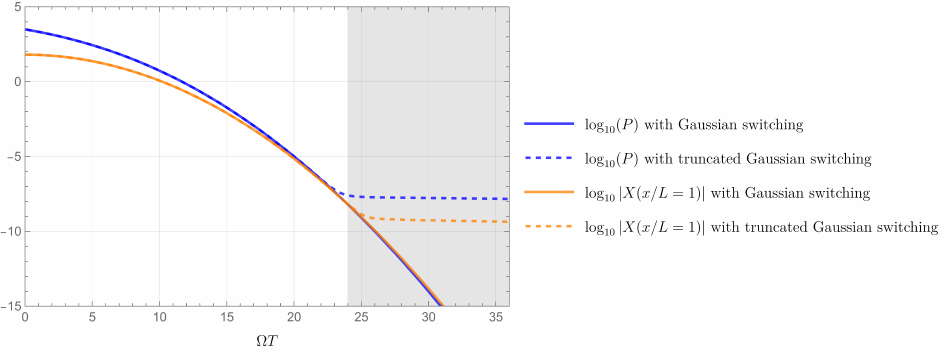}
    \caption{Comparison of multi-detector harvesting with truncated and standard Gaussian switching. The upper panel shows the leading-order negativity for the optimal arrangements with different numbers of detectors. The light gray region shows the nonzero harvesting regime for two detectors with Gaussian switching, where $|X|>P$. This condition never occurs for truncated Gaussian, as seen in the lower panel showing the density matrix elements.}
    \label{fig:truncated_gaussian}
\end{figure*}

The simplest way to enforce strict compact support is to truncate the Gaussian switching in Eq.~\eqref{gaussian_switching} by imposing
a hard cutoff,
\begin{align*}
    \chi(t) &=\frac{1}{\sqrt{2\pi \sigma^2}} e^{-\frac{t^2}{2\sigma^2}}\,\Theta(T-|t|),
\end{align*}
where the tails are removed at $T = 5\sigma$ by multiplication with the Heaviside step function $\Theta$. For a clearer comparison, we do not renormalize the truncated Gaussian\footnote{Moreover, the area within $5\sigma$ of a Gaussian already accounts for approximately $99.9999~\%$ of the total area.}.

The truncation introduces non-smooth behavior at $t=\pm T$. While the Gaussian is $C^\infty$, multiplying it by $\Theta(T-|t|)$ makes the switching only piecewise smooth, with discontinuities at the cutoff points.

As shown in Fig.~\ref{fig:truncated_gaussian}, for small energy gaps the truncated Gaussian switching closely matches the $\tilde{\rho}_1$ elements (in Eq.~\eqref{rho1_pt}) and negativity of the standard Gaussian. However, as the gap increases, these elements deviate significantly, leading to noticeable differences in the resulting negativity.

Our results show that with two detectors, the truncated Gaussian does not harvest any entanglement at leading order, as shown in the upper panel of Fig.~\ref{fig:truncated_gaussian}. As described in Sec.~\ref{sec:two-detect}, two detectors harvest non-zero entanglement if and only if $|X|>P$. For the standard Gaussian this condition is satisfied when the energy gap is large enough, as shown in Fig.~\ref{fig:two-detectors}. In contrast, as shown in the lower panel of Fig.~\ref{fig:truncated_gaussian}, truncating the Gaussian causes the matrix elements $|X|$ and $P$ to plateau before the condition $|X|>P$ is reached.

When using more detectors, we can harvest entanglement in a way that is very similar to the standard Gaussian case. For three and four detectors, we find the optimal spatial arrangements to be the ABA configuration (Fig.~\ref{fig:ABA_arrangement}) and the diagonal square (Fig.~\ref{fig:diagonal_square_arr}), with similar optimal energy gap values, as shown in Fig.~\ref{fig:truncated_gaussian}. The main difference again appears at large energy gaps: instead of a smoother decay as $\Omega$ increases, the leading-order negativity drops to zero soon after the optimal energy gap values.

\subsubsection{Compactified polynomial switching}

\begin{figure*}[h]
    \centering
    \includegraphics[width=0.8\textwidth]{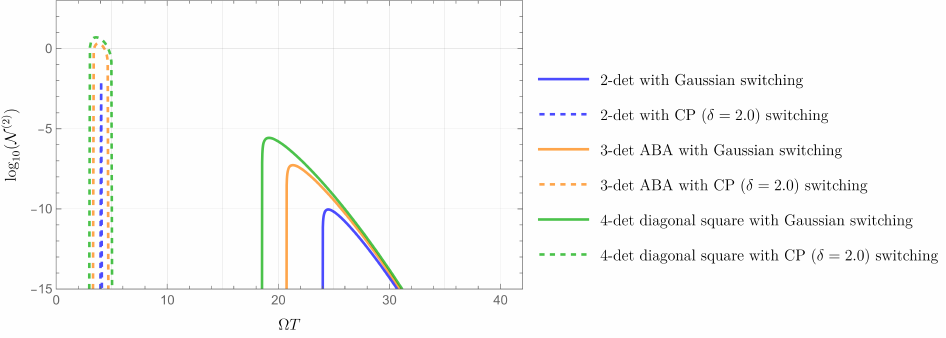}
    \includegraphics[width=0.8\textwidth]{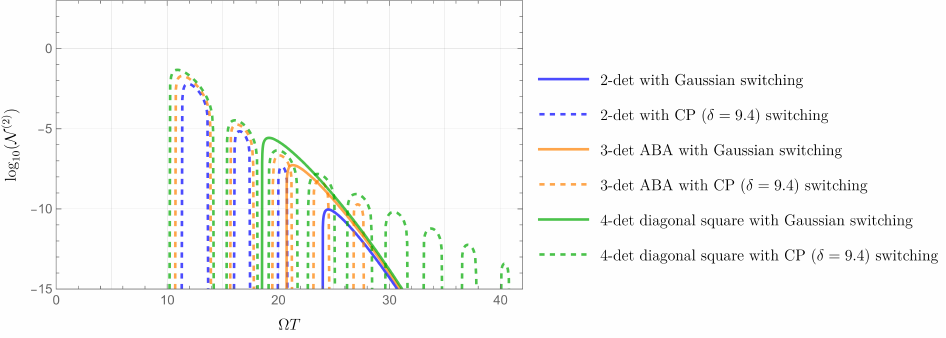}
    \includegraphics[width=0.8\textwidth]{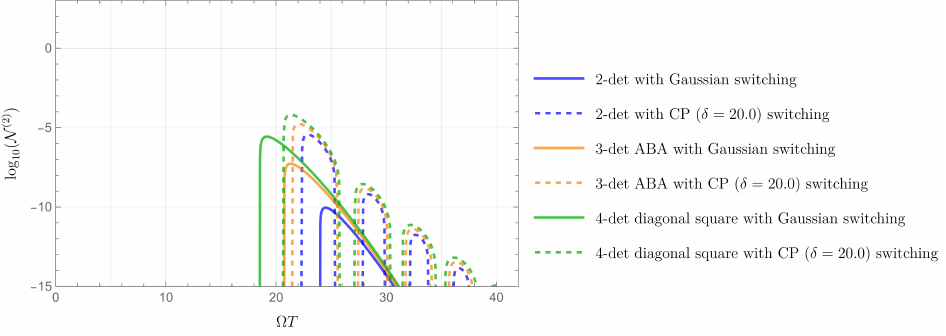}
    \caption{Comparison of multi-detector harvesting using compactified polynomial (CP) and Gaussian switching, for optimal arrangements. Leading-order negativity is plotted as a function of the energy gap, with different panels showing CP results for various differentiability classes, determined by the parameter $\delta$ in Eq.~\eqref{cp_switching}. CP switching produces periodic harvesting regimes, with low differentiability yielding strong entanglement at small gaps and higher differentiability giving more consistent harvesting at larger gaps. Adding detectors expands the harvesting range in energy gap more for Gaussian than CP switching, leading to larger entanglement increases.}
    \label{fig:cp_comparison}
\end{figure*}

In this subsection we introduce a different family of compactly supported polynomial functions with differentiability class  $C^{\delta-1}$ (see \cite{ubiquitous}),
\begin{align}
    \chi^{(\delta)}(t) &= N_\text{norm}\left( 1-\frac{t^2}{T^2} \right)^{\delta}\Theta(T-|t|). \label{cp_switching}
\end{align}
The parameter $\delta$ controls the differentiability at the boundary $t=\pm T$: larger $\delta$ means the cutoff is smoother. In our study, we mainly consider three representative cases:
\begin{itemize}
    \item Minimal value allowing non-zero two-detector harvesting: $\delta=2.0$,
    \item The value that minimizes the relative $L_2$ error with respect to the Gaussian (for $T=5\sigma$): $\delta=9.4$,
    \item A highly differentiable case representing very smooth behavior: $\delta = 20.0$.
\end{itemize}

With compactified polynomial switching, the optimal spatial arrangements are again the ABA configuration (Fig.~\ref{fig:ABA_arrangement}) and the diagonal square (Fig.~\ref{fig:diagonal_square_arr}). This shows that the optimal arrangements are not specific to Gaussian switching functions, although compactified polynomial switching produces clear differences in entanglement harvesting.

The density-matrix elements of $\tilde\rho_1$ generally decrease with increasing energy gap, so smaller gaps harvest more entanglement, similar to the Gaussian switching. On a finer scale, the elements oscillate, producing periodic regions of nonzero harvesting (see Fig.~\ref{fig:cp_comparison}) not seen with Gaussian switching.

Figure \ref{fig:cp_comparison} shows that the two-detector harvesting condition, $|X|>P$, is satisfied periodically as the energy gap increases. Adding a third and fourth detector widens the parameter regimes and results in more consistent periodical harvesting regimes. However, the energy-gap range widens less than in the Gaussian case, so negativity does not appear to scale exponentially when additional detectors are added with this switching.

With compactified polynomial switching, we can analyze how increasing the differentiability class affects the harvesting. The diagonal square configuration begins harvesting at $\delta=1.0$, while the two-detector condition $|X|>P$ is first satisfied at $\delta=2.0$. For these smallest values of $\delta$, there is initially only a single compact region of nonzero harvesting as a function of the energy gap. However, as shown in Fig.~\ref{fig:cp_comparison}, increasing the differentiability class shifts the harvesting regimes towards higher energy gaps and produces new regimes that appear periodically at higher gaps. Because the strongest harvesting occurs at low energy gaps, low $\delta$ is optimal, but it is also more sensitive to the exact value of the energy gap.

\section{Conclusions}\label{sec: conclusions}

In this work, we have examined a system of $N=N_A+N_B$ Unruh-DeWitt detectors weakly coupled to a scalar field and analyzed entanglement harvesting between two subsystems: $A$, containing $N_A$ detectors, and $B$, containing $N_B$ detectors. We quantified the entanglement harvested between the subsystems using the leading-order negativity, which probes NPT entanglement, and showed that it is additive for product states at this order. In this setup, the leading-order negativity is fully determined by a submatrix supported on the one-excitation subspace---generalizing the result of \cite{Kukita-Nambu-2017} via a matrix formulation that imposes no approximations on detector separations. Since the size of this submatrix grows only linearly with the number of detectors, rather than exponentially as in the case of the full reduced density matrix, this result significantly simplifies the analysis of multi-detector bipartite entanglement harvesting.

The central objective of the paper has been to investigate how the spatial arrangement of multiple detectors influences the entanglement that can be harvested between causally disconnected regions. To this end, we explored a variety of detector configurations, increasing the number of detectors in each subsystem and distributing them in space in different ways.

For a system of three detectors, we derived analytic expressions for the negativity for arbitrary spatial placements and identified the configuration that maximizes entanglement: a linear ABA arrangement of the detectors. This result is consistent with—and extends—the tripartite analysis of \cite{threeHarvesting2022}. For four detectors, we systematically compared several symmetric geometries and obtained analytic formulas for the negativity in each case. After optimizing over separations and energy gaps, we found that the diagonal-square configuration yields the greatest amount of harvested entanglement. Additionally, we studied a linear chain with alternating subsystem detectors, showing that the harvested entanglement grows linearly with the number of detectors. As more detectors are added, new eigenvalues of the partially transposed density matrix contribute to the negativity.

A key conclusion emerging from our analysis is that optimal entanglement harvesting is achieved when the separation between detectors belonging to different subsystems is minimized, while the separation between detectors within the same subsystem is maximized. This behavior can be interpreted as a manifestation of entanglement monogamy: strengthening entanglement within a subsystem reduces the ability to generate entanglement across the bipartition. We also observed that increasing the number of detectors broadens the parameter ranges—both in energy gaps and in spatial separations—over which entanglement can be successfully extracted from the field.

Finally, we compared our results using different types of compactly supported switching functions. Two points stand out: the optimal spatial configurations for entanglement harvesting remain unchanged, and the nonzero harvesting energy-gap regimes depend strongly on both the switching function and the arrangement. Truncated Gaussian switchings produce multi-detector harvesting similar to the Gaussian case for small energy gaps, while compactified polynomial switching alters the harvesting profile: low differentiability gives strong harvesting at small energy gaps, while higher differentiability produces periodic but more consistent harvesting regimes at larger gaps. 

We hope that the results presented here will open the door to new investigations in relativistic quantum information processing with multiple detectors. While our analysis has focused on the relatively simple setting of Minkowski spacetime and bipartite entanglement harvesting, future work may extend these techniques to more general curved spacetimes, as well as to genuinely multipartite harvesting scenarios. Furthermore, this work may provide guidance for future experimental realizations of the entanglement harvesting protocols, in setups involving multiple detectors.

\section*{Code availability}

The Mathematica notebooks used for the analysis of Sec.~\ref{sec:Searching for the optimal spatial configurations} are publicly available at:
\href{https://github.com/sansalv/bipartite-entanglement-harvesting-with-multiple-detectors}{GitHub repository}. A citable, archived version of the notebooks is available via Zenodo DOI~\cite{Salomaa2026notebooks}.

\begin{acknowledgments}
The authors would like to thank Eduardo Martín-Martínez for his helpful comments and valuable feedback. SS, SNG and EKV acknowledge the financial support of the Research Council of Finland through the Finnish Quantum Flagship project (358878, UH). SS also acknowledges the financial support of the Finnish Ministry of Education and Culture through the Quantum Doctoral Education Pilot Program (QDOC VN/3137/2024-OKM-4). EKV is also in part supported by the Research Council of Finland grant 1371600.
\end{acknowledgments}

\appendix

\begin{widetext}
\section{Proof of the leading-order additivity of negativity}\label{App:proof_of_additivity}
\begin{proposition}
Let
$$
\hat\rho_i = \hat\rho^{(0)}_i + \lambda^{n_i} \hat\rho^{(n_i)}_i + \mathcal{O}(\lambda^{n_i+1})
$$
be density operators on finite-dimensional bipartite Hilbert spaces such that the unperturbed states satisfy the PPT condition $\big(\hat\rho_i^{(0)}\big)^{T_B} \ge 0$. For $n=\min_i n_i$ and sufficiently small $\lambda$, the leading-order negativity is additive
$$
\mathcal{N}\left( \bigotimes_ {i=1}^N\hat\rho_{i} \right) = \lambda^n \sum_ {i=1}^N \mathcal{N}^{(n)}(\hat\rho_i)+\mathcal{O}(\lambda^{n+1}).
$$
\end{proposition}

\begin{proof}
The proof proceeds as follows: we first express the negativity of a product state exactly in terms of its subsystems. Then we show that only eigenvalues in the kernel of $\hat\rho_i^{(0)\,T_B}$ can become negative under the perturbation, scaling as $\mathcal{O}(\lambda^{n_i})$. Substituting these into the exact formula gives the final additive expression for the product state.

Using the negativity definition in Eq.~\eqref{neg_definition}, together with the locality of the partial transpose and the multiplicativity of the trace norm, one can easily obtain for a strict product state
\begin{align}
    \mathcal{N}\left(\bigotimes_ {i=1}^N\hat\rho_{i}\right)
    &= \frac12 \Big(\big\|(\otimes_ {i=1}^N\hat\rho_{i}\big)^{T_B}\big\|_1-1\Big) \nonumber = \frac12 \Big(\big\|\otimes_ {i=1}^N(\hat\rho_{i}^{T_B})\big\|_1-1\Big) \nonumber \\
    &= \frac12 \Big(\prod_ {i=1}^N\|\hat\rho_{i}^{T_B}\|_1-1\Big) \nonumber \\
    &= \frac{1}{2} \Big[ \prod_{i=1}^N \big( 2 \mathcal{N}(\hat \rho_i) + 1 \big) - 1 \Big]. \label{negativity_exact}
\end{align}
This identity is exact.

Next, assume a perturbative expansion
$$
\hat\rho_i = \hat\rho^{(0)}_i + \lambda^{n_i} \hat\rho^{(n_i)}_i + \mathcal{O}(\lambda^{n_i+1}),
$$
where linearity of the partial transpose gives
$$
\hat\rho_i^{T_B} = \hat\rho_i^{(0)\,T_B} + \lambda^{n_i} (\hat\rho_i^{(n_i)})^{T_B} + \mathcal O(\lambda^{n_i+1}).
$$
Let $\alpha_k(\cdot)$, $k=1,\dots,d$, denote the eigenvalues in non-decreasing order. Using the spectral stability corollary of Weyl’s inequality\footnote{\parbox{\textwidth}{
For Hermitian matrices $A$ and $B$, Weyl’s inequality implies that the ordered eigenvalues satisfy
$|\alpha_k(A+B)-\alpha_k(A)| \le \|B\|_{\mathrm{op}}$
for all $k$, showing that each eigenvalue is Lipschitz continuous with respect to the operator norm \cite{HornJohnson2013}.
}}, we have
$$
|\alpha_k(\hat\rho_i^{T_B}) - \alpha_k(\hat\rho_i^{(0)\,T_B})| \le \|\hat\rho_i^{T_B} - \hat\rho_i^{(0)\,T_B}\|_{\mathrm{op}} = |\lambda|^{n_i} \| (\hat\rho^{(n_i)\,}_i)^{T_B}\|_{\mathrm{op}} + \mathcal O(\lambda^{n_i+1}) = \mathcal O(\lambda^{n_i}),
$$
where $\|\cdot\|_\mathrm{op}$ is the operator (spectral) norm. Hence every eigenvalue can move only by $\mathcal O(\lambda^{n_i})$. In particular, any strictly positive eigenvalue of $\hat\rho_i^{(0)\,T_B}$ remains positive for sufficiently small $\lambda$, so only eigenvalues in the kernel $\ker(\hat\rho_i^{(0)\,T_B})$ can become negative under the perturbation. Let $P_0$ denote the orthogonal projector onto $\ker(\hat\rho_i^{(0)\,T_B})$. By standard degenerate perturbation theory \cite{Sakurai2017} the leading-order shifts of these degenerate eigenvalues are given by the eigenvalues of the projected operator $P_0\,(\hat\rho_i^{(n_i)})^{T_B}\, P_0$. Thus
$$
\alpha_k(\hat\rho_i^{T_B})=\lambda^{n_i}\mu_k+\mathcal O(\lambda^{n_i+1}),
$$
where $\mu_k$ are the eigenvalues of $P_0\,(\hat\rho_i^{(n_i)})^{T_B}\,P_0$. Consequently, any negative eigenvalues, and therefore the negativity $\mathcal{N}(\hat\rho_i)$, appear at order $\lambda^{n_i}$.

Let $n=\min_i n_i$ be the global leading order and $\mathcal{N}^{(n)}(\hat\rho_i)$ the corresponding $\lambda^n$-order coefficient of $\mathcal{N}(\hat\rho_i)$. Then
\begin{align*}
    \mathcal{N}(\hat\rho_i) = \lambda^{n} \mathcal{N}^{(n)}(\hat\rho_i) + \mathcal{O}(\lambda^{n+1}),
\end{align*}
where $\mathcal{N}^{(n)}(\hat\rho_i)=0$ whenever $n_i>n$. Expanding the exact product formula, Eq.~\eqref{negativity_exact}, to leading order in $\lambda$ gives
\begin{align}
    \mathcal{N}\left( \bigotimes_ {i=1}^N\hat\rho_{i} \right) = \lambda^n \sum_ {i=1}^N \mathcal{N}^{(n)}(\hat\rho_i)+\mathcal{O}(\lambda^{n+1}). \label{negativity_additivity}
\end{align}
Only subsystems with $n_i=n$ contribute at leading order. Hence, under the natural perturbative assumptions, the leading-order negativity is additive.
\end{proof}

\section{Closed forms of the symmetric and anti-symmetric terms $C^{\pm}_{ij}$ and $X^{\pm}_{ij}$ } \label{app:computation_of_CX_terms}

In this section, we split the reduced density matrix elements $C_{ij}$ (Eq.~\eqref{C_term}) and $X_{ij}$ (Eq.~\eqref{X_term}) to the symmetric, $C_{ij}^{+}, X_{ij}^{+}$, and anti-symmetric $C_{ij}^{-}, X_{ij}^{-}$ parts by decomposing the Wightman function as
\begin{align}
    W(x,x') = \langle \hat\phi(x)\hat\phi(x')\rangle = \underbrace{ \tfrac12\langle\{\hat\phi(x),\hat\phi(x')\}\rangle}_\text{symmetric (+)} 
    + \underbrace{\tfrac12\langle[\hat\phi(x),\hat\phi(x')]\rangle}_\text{antisymmetric (-)} . \label{commutator_anticommutator}
\end{align}
and compute these terms in closed form.

We start with the matrix element
\begin{align*}
    C_{ij} &:= \lambda^2\int d^4x\,d^4x'\,\Lambda_i(x)\Lambda_j(x') e^{-i(\Omega_i t-\Omega_j t')} W(x,x') \\
    &= \lambda^2\int dt\,dt'\,\chi_i(t)\chi_j(t') e^{-i(\Omega_i t-\Omega_j t')} \langle \phi[f_i](t)\,\phi[f_j](t') \rangle
\end{align*}
where we have defined spacetime smearing $\Lambda_i(x):=\chi_i(t) f_i(\mathbf{x})$ and spatially smeared two-point function
\begin{align*}
    \langle \phi[f_i](t)\,\phi[f_j](t') \rangle := \int d^3x\,d^3x'\,f_i(\mathbf{x})f_j(\mathbf{x}') \langle \phi(\mathbf{x},t)\,\phi(\mathbf{x}',t) \rangle .
\end{align*}
The decomposition of the Wightman function in Eq.~\eqref{commutator_anticommutator} gives $C_{ij}=C^{+}_{ij}+C^{-}_{ij}$, respectively. For a massless scalar field with $\omega_k=|\mathbf{k}|\equiv k$, identical smearing profiles $f_i(\mathbf{x})\equiv f(\mathbf{x} -\mathbf{x}_i)$ and using the Fourier transform convention $\tilde g(\mathbf{k})= \int d^d x \,g(\mathbf{x}) e^{-i\mathbf{k}\cdot\mathbf{x}}$, we can write
\begin{align}
C_{ij}^{+} &:= \frac{\lambda^2}{2}\int dt\,dt'\,\chi_i(t)\chi_j(t') e^{-i(\Omega_i t-\Omega_j t')} \langle \{ \phi[f_i](t)\,\phi[f_j](t') \}\rangle  \\
&= \frac{\lambda^2}{2 x_{ij}}\int_0^\infty  \frac{d k}{(2\pi)^2} \, |f(\mathbf{k})|^2 \sin (k x_{ij}) \int_{-\infty}^{\infty}\int_{-\infty}^{\infty} dt \,dt' \chi_i(t) \chi_j(t')e^{- i(\Omega_i t-\Omega_j t')}   (e^{-i k(t-t')}+e^{i k(t-t')})\\
&=\frac{\lambda^2}{2(2\pi)^2  x_{ij}}\int_0^\infty d \mathcal{K}  \,\Big[  \tilde{\chi}_i(k+\Omega_i) \tilde{\chi}_j^*(k+\Omega_j) + \tilde{\chi}^*_i(k-\Omega_i) \tilde{\chi}_j(k-\Omega_j) \Big]  ,
\end{align}
where we have defined $x_{ij}:=|\mathbf{x}_i-\mathbf{x}_j|$ and $ d \mathcal{K} :=  dk\,| \tilde{f}(k)|^2 \sin(k x_{ij})$. In the last line we further assumed that the smearing functions are spherically symmetric, i.e. $f_i(\mathbf{x})\equiv f(|\mathbf{x} -\mathbf{x}_i|)$. For pointlike detectors with identical energy gaps $\Omega_i\equiv\Omega$ and Gaussian switching functions as in Eq.~\eqref{gaussian_switching}, the expression above can be integrated analytically to obtain
\begin{align}
    C_{ij}^+ &= \frac{\lambda^2}{8\pi^{3/2}x_{ij}\sigma}e^{-\big(\tfrac{x_{ij}}{2\sigma}\big)^2}\text{Im}\left[ e^{i\Omega x_{ij}}\operatorname{erf}(i\tfrac{x_{ij}}{2\sigma}+\Omega\sigma) \right] ,
\end{align}
where $\operatorname{erf}(x)$ corresponds to the error function defined in footnote~\ref{footnote_erf_Daw_func}. 

Following the same procedure and assumptions, the antisymmetric term $C^{-}_{ij}$ is obtained as
\begin{align}
C_{ij}^{-} &:= \frac{\lambda^2}{2}\int dt\,dt'\,\chi_i(t)\chi_j(t') e^{-i(\Omega_i t-\Omega_j t')} \langle [ \phi[f_i](t)\,\phi[f_j](t') ]\rangle \nonumber\\
&=\frac{\lambda^2}{2(2\pi)^2  x_{ij}}\int_0^\infty d \mathcal{K}  \, \Big[  \tilde{\chi}_i(k+\Omega_i) \tilde{\chi}_j^*(k+\Omega_j) - \tilde{\chi}^*_i(k-\Omega_i) \tilde{\chi}_j(k-\Omega_j) \Big] .
\end{align}
For pointlike detectors with Gaussian switching functions as in Eq.~\eqref{gaussian_switching}, the integral can be evaluated analytically, yielding
\begin{align}
    C_{ij}^- &= -\frac{\lambda^2}{8\pi^{3/2}x_{ij}\sigma}\,e^{-\big(\tfrac{x_{ij}}{2\sigma}\big)^2}\sin(\Omega x_{ij}) .
\end{align}
The $P$ terms can be easily computed from the above ones by taking the limit $ P = \lim _{x_{ij} \to 0} C_{ij}$.

The $X$ term given by
\begin{align*}
    X_{ij} &=-\lambda^2 \int dt\,dt'\,\chi_i(t)\chi_j(t') e^{i(\Omega_i t+\Omega_j t')}  \Big(\Theta(t-t')\langle \phi[f_i](t)\phi[f_j](t') \rangle+ \Theta(t'-t)\langle \phi[f_j](t')\phi[f_i](t) \rangle\Big)
\end{align*}
can be decomposed similarly. Using Eq.~\eqref{commutator_anticommutator} together with the identity $\Theta (t-t') + \Theta (t'-t) = 1$, we obtain the symmetric term as
\begin{align}
    X_{ij}^{+} &=-\frac{\lambda^2}{2} \int dt\,dt'\,\chi_i(t)\chi_j(t') e^{i(\Omega_i t+\Omega_j t')}  \Big(\Theta(t-t')\langle \{\phi[f_i](t),\phi[f_j](t') \} \rangle+ \Theta(t'-t)\langle \{\phi[f_j](t'),\phi[f_i](t)  \}\rangle\Big)\nonumber\\
    &= -\frac{\lambda^2}{2} \int dt\,dt'\, \chi_i(t)\chi_j(t') e^{i(\Omega_i t+\Omega_j t')}  \langle \{\phi[f_i](t),\phi[f_j](t') \} \rangle \nonumber\\
    &=-\frac{\lambda^2}{2(2\pi)^2x_{ij}}\int_0^\infty d \mathcal{K}  \, \Big[ \tilde{\chi}_i(k-\Omega_i)\tilde{\chi}^*_j(k+\Omega_j) + \tilde{\chi}_j(k-\Omega_j)\tilde{\chi}^*_i(k+\Omega_i)  \Big] .
\end{align}
where in the second equality we have used the symmetry of the anticommutator. For identical pointlike detectors with Gaussian switching functions, the integral evaluates as
\begin{align}
     X_{ij}^{+} &=-\frac{\lambda^2}{4\pi^2 x_{ij}\sigma}\,e^{-\sigma^2\Omega^2}\,D\left(\frac{ x_{ij}}{2\sigma}\right) ,
\end{align}
where $D(x)$ corresponds to the Dawson function as defined in footnote~\ref{footnote_erf_Daw_func}.\\

Finally, the more intricate term to evaluate is $X^{-}_{ij}$, given by
\begin{align*}
    X_{ij}^{-} &=-\frac{\lambda^2}{2} \int dt\,dt'\, \chi_i(t)\chi_j(t') e^{i(\Omega_i t+\Omega_j t')}  \Big(\Theta(t-t')\langle [ \phi[f_i](t),\phi[f_j](t') ] \rangle+ \Theta(t'-t)\langle [\phi[f_j](t')\,\phi[f_i](t)  ] \rangle\Big)\\
    &=-\frac{\lambda^2}{2} \int dt\,dt'\,\chi_i(t)\chi_j(t') e^{i(\Omega_i t+\Omega_j t')}  \Big( \langle [\phi[f_j](t'),\phi[f_i](t)  ] \rangle  + 2\, \Theta(t-t')\langle [ \phi[f_i](t),\phi[f_j](t') ] \rangle \Big).
\end{align*}
This term splits in two different integrals. The first term is
\begin{align}
    I_1 &= -\frac{\lambda^2}{2} \int dt\,dt'\, \chi_i(t)\chi_j(t') e^{i(\Omega_i t+\Omega_j t')}  \langle[\phi[f_j](t'),\phi[f_i](t)  ] \rangle \nonumber\\
    &=-\frac{\lambda^2}{2(2\pi)^2x_{ij}}\int_0^\infty d \mathcal{K}  \,\Big[ \tilde{\chi}_i(k-\Omega_i)\tilde{\chi}^*_j(k+\Omega_j) - \tilde{\chi}_j(k-\Omega_j)\tilde{\chi}^*_i(k+\Omega_i)  \Big] ,
\end{align}
which, for identical switching functions $\chi_i(t)\equiv\chi(t)$ and equal energy gaps $\Omega_i\equiv \Omega$, reduces to  $I_1\equiv 0$, as in all the cases considered in this article. The second term contributing to $X_{ij}^{-}$ is given by
\begin{align*}
    I_2 &= - \lambda^2\int dt\,dt'\, \chi_i(t)\chi_j(t') e^{i(\Omega_i t+\Omega_j t')}   \, \Theta(t-t')\langle [ \phi[f_i](t),\phi[f_j](t') ] \rangle \\
    &= -\frac{\lambda^2}{2(2\pi)^2x_{ij}}\int_0^\infty d \mathcal{K}  \int_{-\infty}^{\infty}\int_{-\infty}^{\infty} dt\,dt'\, \chi_i(t)\chi_j(t') e^{i(\Omega_i t+\Omega_j t')}   \, \Theta(t-t') \Big( e^{- i k (t-t')} - e^{ i k (t-t')}\Big).
\end{align*}
To compute this integral, it is useful to use the following integral representation of the Heaviside step function,
\begin{align}
    \Theta(z)=\lim_{\varepsilon\to 0^+}\frac{i}{2\pi} \int_ {-\infty}^\infty d\tau\,\frac{e^{-iz\tau}}{\tau + i\varepsilon} \label{Heaviside_decomp} .
\end{align}
Using this form, the integral $I_2$ can be written as
\begin{align}
    I_2 &= \lim_{\varepsilon\to 0^+}\frac{\lambda^2}{2(2\pi)^3x_{ij}} \int_0^\infty d \mathcal{K} \int_{-\infty}^\infty\frac{d\tau}{\tau + i \varepsilon}\int_{-\infty}^\infty\int_{-\infty}^\infty   \, dt \,dt' \, \chi_i(t)\chi_j(t') e^{i(\Omega_i t+\Omega_j t')} \Big( e^{- i  (k+\tau) (t-t')} - e^{ i (k-\tau)(t-t')}\Big)\nonumber\\
    &= \lim_{\varepsilon\to 0^+} \frac{\lambda^2}{2(2\pi)^3x_{ij}}\int_0^\infty d \mathcal{K}\int_{-\infty}^{\infty} \frac{d\tau}{\tau + i \varepsilon} \Big[ \tilde{\chi}_i(k +\tau-\Omega_i)\tilde{\chi}^*_j(k+\tau +\Omega_j) - \tilde{\chi}_j(k-\tau-\Omega_j)\tilde{\chi}^*_i(k-\tau +\Omega_i)  \Big] .
\end{align}
This integral can be evaluated analytically for the identical pointlike detector configuration with identical switching functions used throughout the paper. In this case, $I_1\equiv0$, and therefore
\begin{align}
    X^{-}_{ij}\equiv I_2=\frac{i\lambda^2}{8\pi^{3/2}x_{ij}\sigma}\,e^{-\sigma^2\Omega^2-\big(\tfrac{x_{ij}}{2\sigma}\big)^2} .
\end{align}

\section{Construction of $\tilde{\rho}_1$ for the alternative linear chain}\label{app:linear_chain_rho1}

In this short section, we explain how to construct the one-excitation block $\tilde{\rho}_1$ in Eq.~\eqref{rho1_pt_dirac} for the alternative linear chain (Fig.~\ref{fig:appendix_linear_chain}) introduced in Sec.~\ref{sec:Alternative linear chain}, consisting of $N$ identical detectors. This block is required for the computation of the leading-order negativity.

\begin{figure}[H]
    \centering
    \begin{tikzpicture}
        % Define nodes at specific coordinates
        \node (1) at (0,0) [draw, circle, fill=white, minimum size=15pt, inner sep=2pt, label=above:{\scriptsize$1$}] {$A$};
        \node (2) at (1.5,0) [draw, circle, fill=gray!50, minimum size=15pt, inner sep=2pt, label=above:{\scriptsize$N_A+1$}] {$B$};
        \node (3) at (3,0) [draw, circle, fill=white, minimum size=15pt, inner sep=2pt, label=above:{\scriptsize$2$}] {$A$};
        \node (4) at (4.5,0) [draw, circle, fill=gray!50, minimum size=15pt, inner sep=2pt, label=above:{\scriptsize$N_A+2$}] {$B$};
        \node (5) at (6,0) {};
        % Draw connecting lines
        \draw[thick] (1) -- (2) -- (3) -- (4) -- (5);     
        % Add side labels
            \node[below] at ($(1)!0.5!(2)$) {\( L \)};
            \node[below] at ($(2)!0.5!(3)$) {\( L \)};
            \node[below] at ($(3)!0.5!(4)$) {\( L \)};
            \node[below] at ($(4)!0.5!(5)$) {\( L \)};
            \node[right] at ($(5)$) {\( \cdots \)};
    \end{tikzpicture}
    \caption{Linear chain setup with alternating subsystem detectors with minimal causal separation $L\equiv 2T$. Indices are shown above each detector.}
    \label{fig:appendix_linear_chain}
\end{figure}
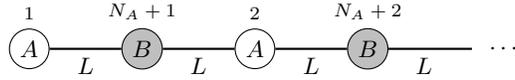

Since the chain starts with an $A$ detector, the detector numbers in the two subsystems are $N_A=\lceil\frac{N}{2}\rceil$ and $N_B=\lfloor\frac{N}{2}\rfloor$. Thus, subsystem $A$ contains one extra detector whenever $N$ is odd. Using the block-diagonal ordering introduced in Sec~\ref{sec:Perturbative evaluation of the reduced density matrix}, the block $\tilde{\rho}_1$ takes the form in Eq.~\eqref{rho1_pt}:
\begin{align*}
    \tilde{\rho}_1 =
    \begin{pmatrix}
    \mathcal{C}_{BB}^* & \mathcal{X}_{BA}^{\dag} \\
    \mathcal{X}_{BA} & \mathcal{C}_{AA}
    \end{pmatrix}.
\end{align*}

Since the detectors are equally spaced by a distance $L$, the blocks $\mathcal{C}_{BB}$, $\mathcal{C}_{AA}$ and $\mathcal{X}_{BA}$ depend only on the separation between detector labels. Taking $P\equiv C_{0L}$ and distance-based element labeling, we obtain the matrix representation as
\begin{align*}
    (\tilde\rho_1)_{ij}=
    \begin{cases}
    C_{2|i-j|L}^* & \text{if }i,j\leq N_B,\\
    C_{2|i-j|L} & \text{if }i,j> N_B,\\
    X_{(2k(i,j)-1)L} & \text{if }i>N_B\text{ and } j\leq N_B,\\
    (\tilde\rho_1)^*_{ji} &\text{if }i\leq N_B\text{ and } j> N_B,
    \end{cases}
\end{align*}
where the index $k(i,j)$ is
\begin{align*}
    k(i,j) = \max\{(i-N_A)-j+1,\,j-(i-N_A)\}.
\end{align*}
Blocks $\mathcal{C}_{AA}$ and $\mathcal{C}_{BB}$ exhibit a Toeplitz structure with entries depending on $|i-j|$. Similarly, $\mathcal{X}_{BA}$ is a Toeplitz matrix with a banded staircase structure, where each diagonal corresponds to a fixed separation $X_{(2k-1)L}$ up to $X_{(2N_B-1)L}$.

\end{widetext}

\bibliographystyle{apsrev4-1}
\bibliography{references.bib}

\end{document}